\documentclass{llncs}
\pdfoutput=1

\usepackage{tikz}
\usetikzlibrary{snakes,arrows,patterns,shadings}
\usepackage{xspace}
\usepackage{macros}
\usepackage{xcolor}
\usepackage[nointegrals]{wasysym}
\usepackage{amsmath,amssymb}
\usepackage{subfigure}
\usepackage{wrapfig}
\usepackage[linesnumbered,procnumbered,ruled,vlined]{algorithm2e}
\usepackage{hyperref}
\usepackage{cleveref,thm-restate}
\usepackage{enumerate}
\usepackage{boxedminipage}
\usepackage{colonequals}
\usepackage{tabularx}
\usepackage{bibunits}
\usepackage{url}

\usepackage{microtype}

\makeatletter
\let\c@lemma\c@theorem
\let\c@corollary\c@theorem
\let\c@definition\c@theorem
\let\c@proposition\c@theorem
\makeatother
 
\begin{document}
\SetKwData{Cost}{{\scshape Cost}}
\SetKwData{Passed}{{\scshape Passed}}
\SetKwData{Waiting}{{\scshape Waiting}}

\begin{bibunit}[alpha]

\title{Symbolic Optimal Reachability in \\ Weighted Timed Automata}

\author{Patricia \textsc{Bouyer} \and Maximilien \textsc{Colange} \and Nicolas \textsc{Markey}}

\institute{LSV -- CNRS -- ENS Cachan} 

\maketitle

\begin{abstract}
  Weighted timed automata have been defined in the early 2000's for
  modelling resource-consumption or -allocation problems in real-time
  systems. 
  Optimal reachability is decidable in weighted timed automata, and a
  symbolic forward algorithm has been developed to solve that
  problem. This algorithm uses so-called \emph{priced zones}, an~extension of
  standard zones with \emph{cost functions}. In~order to ensure termination, 
  the algorithm requires clocks to be bounded.
  For unpriced timed automata,
  much work has been done to develop sound abstractions adapted to the
  forward exploration of timed automata, ensuring termination of the
  model-checking algorithm without bounding the clocks.  
  In~this paper, we~take advantage of recent developments on
  abstractions for timed automata, and propose an algorithm allowing
  for symbolic analysis of all weighted timed automata, without
  requiring bounded clocks.
\end{abstract}

\section{Introduction}
\label{sec:Introduction}

Timed automata~\cite{AD94} have been introduced in the early 1990's as
a powerful model to reason about (the~correctness~of) real-time
computerized systems. Timed automata extend finite-state automata with
several clocks, which can be used to enforce timing constraints
between various events in the system. They~provide a convenient
formalism and enjoy reasonably-efficient algorithms (e.g.~reachability
can be decided using polynomial space), which explains the enormous
interest that they raised in the community of formal verification.

Hybrid automata~\cite{ACHH93} can be viewed as an extension of timed
automata, involving hybrid variables: those variables can be used to
measure other quantities than time (e.g. temperature, energy
consumption,~...). Their evolution may follow differential equations,
depending on the state of the system. Those variables unfortunately
make the reachability problem undecidable~\cite{HKPV98}, even in the
restricted case of stopwatches (i.e.,~clocks that can be stopped
and restarted).

Weighted (or priced) timed automata~\cite{ATP01,BFH+01} have been
proposed in the early 2000's as an intermediary model for modelling
resource-consumption or -allocation problems in real-time systems
(\eg~optimal scheduling~\cite{BLR05b}). \Cref{fig-exintro} displays an
example of a weighted timed automaton, modelling aircrafts~(left) that
have to land on runways~(right).  In~(single-variable) weighted timed
automata, each location carries an integer, which is the rate by which
the hybrid variable (called \emph{cost} variable hereafter) increases
when time elapses in that location. Edges may also carry a value,
indicating how much the cost increases when crossing this edge. Notice
that, as~opposed to (linear) hybrid systems, the constraints on edges
(a.k.a.~\emph{guards}) only involve clock variables:
the~extra quantitative information measured by the \emph{cost} is just an
observer of the system, and it does not interfere with the behaviors of the
system.

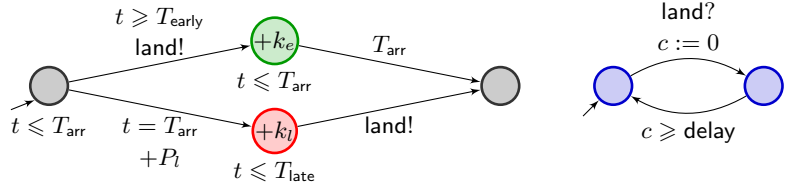
\begin{figure}[t]
\def\arr{\textsf{arr}}
\def\early{\textsf{early}}
\def\late{\textsf{late}}
\def\land{\textsf{land}}
\def\delay{\textsf{delay}}
  \centering
  \begin{tikzpicture}

  \begin{scope}[xshift=5cm]    
    \draw (0,0) node[rond,colbleu] (l0) {};
    \draw (2,0) node[rond,colbleu] (l1) {};
    \draw (l0) edge[draw,line width=.4pt,-latex',bend left] 
      node[above] {$\genfrac{}{}{0pt}{0}{\land?}{c:=0}$} (l1);
    \draw (l1) edge[draw,line width=.4pt,-latex',bend left] 
      node[below] {$c\geq \delay$} (l0);
    \draw[latex'-] (l0.-135) -- +(-135:3mm);
    \end{scope}

    \begin{scope}[xshift=-2.5cm,yscale=.6,xscale=1.5]
    \draw (0,0) node [rond,colgris] (l'0) {} 
      node[below=3mm] {$t\leq T_{\arr}$};    
    \draw[latex'-] (l'0.-135) -- +(-135:3mm);
    \draw (2,1) node [rond,colvert,inner sep=0pt] (l'1) {$+k_e$} 
      node[below=3mm] {$t\leq T_{\arr}$};    
    \draw (2,-1) node [rond,colrouge,inner sep=0pt] (l'2) {$+k_l$} 
      node[below=3mm] {$t\leq T_{\late}$};    
    \draw (4,0) node [rond,colgris] (l'3) {};    

    \draw (l'0) edge[draw,line width=.4pt,-latex'] 
      node[above] {$\genfrac{}{}{0pt}{0}{t\geq T_{\early}}{\land!}$} (l'1);
    \draw (l'0) edge[draw,line width=.4pt,-latex'] 
      node[below] {$\genfrac{}{}{0pt}{0}{t=T_{\arr}}{+P_l}$} (l'2);
    \draw (l'1) edge[draw,line width=.4pt,-latex'] 
      node[above] {$T_{\arr}$} (l'3);
    \draw (l'2) edge[draw,line width=.4pt,-latex'] 
      node[below] {$\land!$} (l'3);
    \end{scope}

    \end{tikzpicture} 
  \caption{A (simplified) model of the Aircraft Landing System~\cite{LBB+01}:
  aircrafts (left) have an optimal landing time~$T_\arr$ within a possible landing
  interval $[T_{\early},T_{\late}]$. The aircraft can speed up (which incurs
  some extra cost, modelled by~$k_e$) to land earlier than~$T_{\arr}$, or can
  delay landing (which also entails some penalties, modelled by~$P_l$
  and~$k_l$). Some delay has to occur between consecutive landings on the same
  runway, because of wake turbulence; this is taken into account by the model
  of the runways (right).}
\label{fig-exintro}
\end{figure}

Optimal cost for reaching a target, and associated almost-optimal schedules,
can be computed in weighted timed automata~\cite{ATP01,BFH+01,BBBR07}. The
proofs of these results rely on region-based algorithms (either priced
regions~\cite{BFH+01}, or corner-point refinements~\cite{ATP01,BBBR07}).
Similarly to standard regions for timed automaton~\cite{AD94}, such
refinements of regions are not adapted to a real implementation. A~symbolic
approach based on priced zones has been proposed in~\cite{LBB+01}, and later
improved in~\cite{RLS06}. Zones are a standard symbolic representation for
the analysis of timed-automata~\cite{BY03,bouyer04}, and priced zones extend
zones with cost functions recording, for each state of the zone, the
optimal cost to reach that state. 
A~forward
computation in a weighted timed automaton can be performed using priced
zones~\cite{LBB+01}: 
it~is based on a single-step \Post-operation on priced zones, and on
a basic inclusion test between priced zones 
(inclusion of zones, and point-to-point
comparison of the cost function on the smallest zone).
The~algorithmics has been improved in~\cite{RLS06}, and termination and
correctness of the forward computation is obtained for weighted timed
automata \emph{in which all clocks are bounded}. 
Bounding clocks of a weighted timed automaton can always be achieved (while
preserving the cost), but it may increase the size of the model.
We~believe that a better
solution is possible: for timed automata and zones, a~lot of
efforts have been put into the development of sound abstractions adapted to the
forward exploration of timed automata, ensuring termination of the
model-checking algorithms without bounding
clocks~\cite{BY03,BBFL03,BBLP05,HKSW11,HSW12}.  

\bigskip In this paper, we build on~\cite{LBB+01,RLS06}, and extend
the symbolic algorithm to general weighted timed automata, without
artificially bounding the clocks of the model. The keypoint of our algorithm
is an inclusion test between \emph{abstractions} of priced zones, computable
from the (non abstracted) priced zones themselves. It~can be seen as a priced
counterpart of a recently-developed inclusion test over standard
zones~\cite{HSW12}: it~compares abstractions of zones without explicitly
computing them, which has shown its efficiency for the
analysis of timed automata. We~prove that the forward-exploration algorithm
using priced zones with this inclusion test indeed computes the optimal cost,
and that it terminates.
We~also propose an algorithm to effectively
decide inclusion of priced zones. 
We~implemented our algorithm, and we compare it with that of~\cite{RLS06}. 

\paragraph{Related work.}
The approach of~\cite{LBB+01,RLS06} is the closest related work. Our algorithm
applies to a more general class of systems (unbounded clocks), and always
computes fewer symbolic states on bounded models
(see~Remark~\ref{rmk:inclusion});
also, 
while the inclusion test of~\cite{RLS06} reduces to a mincost
flow problem, for which efficient algorithms exist, 
we~had to develop specific algorithms for checking our new inclusion relation.
We~develop this comparison with~\cite{RLS06} further in~\Cref{sec:exp},
including experimental results.

Our algorithm can be used in particular to compute best- and
worst-case execution times. Several tools propose WCET analysis based
on timed automata: TIMES~\cite{AFM+03} uses binary-search to evaluate
WCET, while Uppaal~\cite{GELP10} and METAMOC~\cite{metamoc}
rely on the algorithm of~\cite{RLS06} mentioned
above; in~particular they require bounded
clocks to ensure termination. A~tentative workaround to this problem has been
proposed 
in~\cite{ARF14}, but we~are uncertain about its correctness (as~we~explain
with a counter-example in~\Cref{app:australiens}).

\medskip
By lack of space, all proofs are given in a separate appendix.

\section{Weighted timed automata}
\label{sec:wta}

In this section we define the weighted (or priced) timed automaton
model, that has been proposed in 2001 for representing resource
consumption in real-time systems~\cite{ATP01,BFH+01}

We consider as time domain the set $\IR_{\geq 0}$ of non-negative
reals. We let $X$ be a finite set of variables, called \emph{clocks}.
A \emph{(clock) valuation} over $X$ is a mapping $v\colon X
\rightarrow \IR_{\geq 0}$ that assigns to each clock a time value. The
set of all valuations over~$X$ is denoted $\IR_{\geq 0}^X$. Let $t \in
\IR_{\geq 0}$, the valuation $v+t$ is defined by $(v+t)(x)= v(x)+t$
for every $x\in X$. For $Y \subseteq X$, we denote by $[Y \leftarrow
0]v$ the valuation assigning $0$ (respectively $v(x)$) to every $x\in
Y$ (respectively $x\in X\setminus Y$). We~write~$\mathbf{0}_X$ for the
valuation which assigns $0$ to every clock $x \in X$.

The set of \emph{clock constraints} over $X$, denoted $\C(X)$, is
defined by the grammar
$g\coloncolonequals x \sim c\ \mid\ g \wedge g$,
where $x \in X$ is a clock, $c \in \IN$, and $\mathord\sim \in
\{\mathord<,\mathord\leq,\mathord=,\mathord\geq,\mathord>\}$.

Clock constraints are evaluated over clock valuations, and the
satisfaction relation, denoted $v \models g$, is defined inductively
by $v \models (x \sim c)$ whenever $v(x) \sim c$, and $v \models g_1
\wedge g_2$ whenever $v \models g_1$ and $v \models g_2$.

\begin{definition}
  A \emph{weighted timed automaton} is a tuple $\A =
  (X,L,\ell_0,\Goal,\penalty1000 E,\weight)$ where $X$ is a finite set of clocks,
  $L$ is a finite set of locations, $\ell_0 \in L$ is the initial
  location, $\Goal \subseteq L$ is a set of goal (or final) locations,
  $E \subseteq L \times \C(X) \times 2^X \times L$ is a finite set of
  edges (or transitions), and $\weight: L \cup E \rightarrow \IZ$ is a
  \emph{weight function} which assigns a value to each location and to
  each transition.
\end{definition}
In the above definition, if we omit the \weight\ function, we
obtain the well-known model of \emph{timed automata}~\cite{AD90,AD94}.
The semantics of a weighted timed automaton is that of the underlying
timed automaton, and 
the \weight\ function provides 
quantitative information about the moves and executions of the system.

The semantics of a timed automaton $\A =
(X,L,\ell_0,\Goal,E)$ 
is given as a timed transition system
$\mathcal{T}_{\A} = (S,s_0,\rightarrow)$ where $S = L \times \IR_{\geq
  0}^X$ is the set of configurations (or states) of $\A$, $s_0 =
(\ell_0,\mathbf{0}_X)$ is the initial configuration, and $\rightarrow$
contains two types of moves:
\begin{itemize}
\item delay moves: $(\ell,v) \xrightarrow{t} (\ell,v+t)$ if
  $t \in \IR_{\geq 0}$;
\item discrete moves: $(\ell,v) \xrightarrow{e} (\ell',v')$
  if there exists an edge $e = (\ell,g,Y,\ell')$ in $E$ such that $v
  \models g$, $v' = [Y \leftarrow 0]v$.
\end{itemize}

A run $\varrho$ in $\A$ is a finite sequence of moves in the
transition system ${\cal T}_{\A}$, with a strict alternation of delay
moves (though possibly $0$-delay moves) and discrete moves. In the
following, we may write a run $\varrho = s \xrightarrow{t_1} s'_1
\xrightarrow{e_1} s_1 \xrightarrow{t_2} s'_2 \xrightarrow{e_2} s_2
\ldots$ more compactly as $\varrho = s \xrightarrow{t_1,e_1} s_1
\xrightarrow{t_2,e_2} s_2 \cdots$. 
If $\varrho$ ends in some $s = (\ell,v)$ with
$\ell \in \Goal$, we say that $\varrho$ is accepting. 
For a configuration $s\in S$, we~write $\Runs(\A,s)$
the set of accepting runs 
that start in $s$.

In the following we will assume timed automata are non-blocking, that
is, from every reachable configuration $s$, there exist some delay
$t$, some edge $e$ and some configuration $s'$ such that $s
\xrightarrow{t,e} s'$ in $\A$.

\medskip We can now give the semantics of a weighted timed automaton
$\A = (X,L,\ell_0,\penalty1000\relax \Goal,E,\weight)$.  The value
$\weight(\ell)$ given to location $\ell$ represents a cost rate, and
delaying $t$ time units in a location $\ell$ will then cost
$t\cdot\weight(\ell)$. The value $\weight(e)$ given to edge $e$
represents the cost of taking that edge.  Formally, the cost of the
two types of moves 
is defined as follows:
\[
\left\{\begin{array}{l} \cost\left((\ell,v)
      \xrightarrow{t} (\ell,v+t)\right) = t \cdot \weight(\ell) \\
    \cost\left((\ell,v) \xrightarrow{e} (\ell',v')\right) = \weight(e)
  \end{array}\right.
\]
A \emph{run} $\varrho$ of a weighted timed automaton is a run of the
underlying timed automaton.
The cost of $\varrho$, denoted $\cost(\varrho)$, is
the sum of the costs of all the simple moves along $\varrho$.

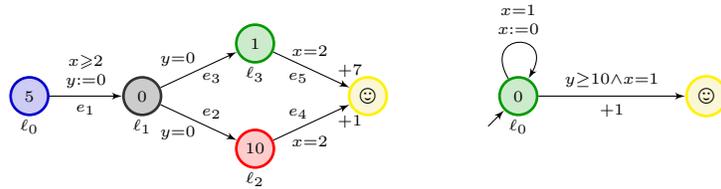
\begin{figure}[ht]
  \centering
  \begin{tikzpicture}
    \everymath{\scriptstyle} 
 \begin{scope}[yscale=.7]
   \draw (-.5,0) node[rond,colbleu,inner sep=0pt] (l0) {$5$} 
     node[below=2mm] {$\ell_0$};
   \draw (1,0) node[rond,colgris,inner sep=0pt] (l1) {$0$} 
     node[below=2mm] {$\ell_1$};
   \draw (2.5,-1) node[rond,colrouge,inner sep=0pt] (l2) {$10$}
      node[below=2mm] {$\ell_2$};
   \draw (2.5,1) node[rond,colvert,inner sep=0pt] (l3) {$1$}
      node[below=2mm] {$\ell_3$};
   \draw (4,0) node[rond,coljaune,inner sep=0pt] (l4) {\smiley} ;
     
   \draw (l0) edge[-latex',line width=.4pt,draw] 
     node[above] {$\genfrac{}{}{0pt}{1}{x\geq 2}{y:=0}$} node[below] {$e_1$}
     (l1); 
   \draw (l1) edge[-latex',line width=.4pt,draw] 
     node[below left=-2pt] {$y=0$} node[above right=-2pt] {$e_2$} (l2); 
   \draw (l1) edge[-latex',line width=.4pt,draw] 
     node[above left=-2pt] {$y=0$} node[below right=-2pt] {$e_3$} (l3); 
   \draw (l2) edge[-latex',line width=.4pt,draw] 
    node[pos=.2,below right=-2pt] {$x=2$} node[pos=.8,below right=-2pt] {$+1$} 
    node[above left=-2pt] {$e_4$} (l4); 
   \draw (l3) edge[-latex',line width=.4pt,draw] 
    node[pos=.2,above right=-2pt] {$x=2$} node[pos=.8,above right=-2pt] {$+7$} 
    node[below left=-2pt] {$e_5$} (l4); 
 \end{scope}

  \begin{scope}[xshift=6cm]
    \draw (0,0) node[rond,colvert,inner sep=0pt] (l0) {$0$} 
      node[below=2mm] {$\ell_0$};
    \draw (2.5,0) node [rond,coljaune,inner sep=2pt] (lf) {$\smiley$};
      \draw [latex'-] (l0.-135) -- +(-135:3mm);
      \draw [-latex'] (l0) .. controls +(120:1cm) and +(60:1cm)
      .. (l0) node [midway,above] {$\genfrac{}{}{0pt}{1}{x=1}{x:=0}$};
      \draw [-latex'] (l0) -- (lf) node [midway,above]
      {$y \ge 10 \wedge x=1$} node [midway,below] {$+1$};

  \end{scope}
    \end{tikzpicture}
  \caption{Examples of weighted timed automata}
  \label{fig:ex1}\label{fig:non-terminate}
\end{figure}
\begin{example}
  \label{ex:1}
  We consider the weighted timed automaton $\A$ depicted in
  \Cref{fig:ex1}~(left). When a weight is non-null, we add a
  corresponding decoration to the location or to the
  transition. A~possible run in~$\A$~is:
  \[
  \varrho\ =\ (\ell_0,0) \xrightarrow{0.1} (\ell_0,0.1)
  \xrightarrow{e_1} (\ell_1,0.1) \xrightarrow{e_3} (\ell_3,0.1)
  \xrightarrow{1.9} (\ell_3,2) \xrightarrow{e_5} (\smiley,2)
  \]
  The cost of $\varrho$ is $\cost(\varrho) = 5 \cdot 0.1 + 1 \cdot 1.9 + 7 =
  9.4$ (the cost per time unit is~$5$ in~$\ell_0$, $1$~in~$\ell_3$, and the
  cost of transition $e_5$ is~$7$).
\end{example}

\subsection*{The optimal-reachability problem}

For this model we are interested in the optimal-reachability problem,
and in the synthesis of almost-optimal schedules.  Given a weighted
timed automaton $\A = (X,L,\ell_0,\Goal,E,\weight)$, the optimal cost
from $s = (\ell,v)$ is defined as:
\[
\optcost_\A(s) = \inf_{\varrho \in \Runs(\A,s)} \cost(\varrho)
\]
If $\epsilon>0$, a run $\varrho \in \Runs(\A,s)$ is
\emph{$\epsilon$-optimal} whenever $\cost(\varrho) \le \optcost_\A(s)
+ \epsilon$.

We are interested in $\optcost_\A(s_0)$, simply written as
$\optcost_\A$, when $s_0$ is the initial configuration of $\A$.  It is
known that $\optcost_\A$ can be computed in polynomial
space~\cite{ATP01,BFH+01,BBBR07}, and that almost-optimal schedules
(that is, for every $\epsilon>0$, $\epsilon$-optimal schedules) can
also be computed.

The solutions developed in the aforementioned papers are based on
refinements of regions, and a symbolic approach has been proposed
in~\cite{LBB+01,RLS06}, which extends standard zones with cost
functions: this algorithm computes the optimal cost in weighted timed
automata with nonnegative weights, assuming the underlying timed
automata are \emph{bounded}, that is, there is a constant~$M$ such
that no clock can go above~$M$.  This is without loss of generality
w.r.t. optimal cost, since any weighted timed automaton can be
transformed into a bounded weighted timed automaton with the same
optimal cost; it may nevertheless increase the size of the model, and
more importantly of the state-space which needs to be explored (it can
be exponentially larger). We believe that a better solution is
possible: for timed automata and zones, a lot of efforts have been put
into the development of sound abstractions adapted to the forward
exploration of timed automata, ensuring termination of the
model-checking algorithm without bounding
clocks~\cite{BY03,BBFL03,BBLP05,HKSW11,HSW12}.

Building on~\cite{LBB+01,RLS06}, we extend the
symbolic algorithm to general weighted timed automata, without
assuming bounded clocks. The keypoint of our
algorithm is an inclusion test between \emph{abstractions} of priced
zones, computable from the (non abstracted) priced zones
themselves. It~can be seen as a priced counterpart of a
recently-developed inclusion test over standard zones~\cite{HSW12},
which compares their abstractions without explicitly computing them,
and has shown its efficiency for the analysis of timed
automata. We~prove that the symbolic algorithm using priced zones and
this inclusion test indeed computes the optimal cost, and that it
terminates.

\section{Symbolic algorithm}
\label{sec:algo}

In this section we briefly recall the approach of~\cite{LBB+01,RLS06},
and explain how we extend it to the general model, explaining which
extra operation is required. The rest of the paper is devoted to
proving correctness, effectiveness and termination of our algorithm.

\subsection{The symbolic representation:
priced zones}

Let $X$ be a finite set of clocks. A \emph{zone} is a set of valuations
defined by a generalized constraint over clocks, given by the grammar
$\gamma~\coloncolonequals~ x \sim c\ \mid\ x-y \sim c \ \mid\ \gamma
\wedge \gamma$,
where $x, y \in X$ are clocks, $c \in \IZ$, and $\mathord\sim \in
\{\mathord<,\mathord\leq,\mathord=,\mathord\geq,\mathord>\}$. 
Zones and their representation using Difference
Bound Matrices (DBMs~in short) are the standard symbolic data
structure used in tools implementing timed
systems~\cite{BY03,bouyer04}.

To deal with weighted timed automata, zones have been extended to
priced zones in~\cite{LBB+01}. A~\emph{priced zone} is a pair
$\mathcal{Z} = (Z,\zeta)$ where $Z$ is a zone, and $\zeta\colon
\IR_{\geq 0}^X \to \IR$ is an affine function. In a symbolic state
$(\ell,\mathcal{Z})$, the cost function $\zeta$ is meant to represent
the optimal cost so far (that is, $\zeta(v)$ is the optimal cost so
far for reaching configuration~$(\ell,v)$).  In~\cite{LBB+01}, it~is
shown how one can simply represent priced zones, and how these can be
used in a forward-exploration algorithm. The algorithm is shown
as~\Cref{algo:optreach}, and we parametrize it by an inclusion
test~$\preceq$ between priced zones.

\begin{algorithm}[b]

\DontPrintSemicolon

$\Cost \gets \infty$ \; $\Passed \gets \emptyset$ \; $\Waiting \gets
\{ (\ell_0, \mathcal{Z}_0) \}$ \; \While{$\Waiting \neq \emptyset$}{
  select $(\ell, \mathcal{Z})$ from $\Waiting$ \; \If{$\ell \in \Goal$
    and $\mincost(\mathcal{Z}) < \Cost$}{ $\Cost \gets
    \mincost(\mathcal{Z})$ \; } 
  \If{for all $(\ell, \mathcal{Z'}) \in \Passed$,
    $\mathcal{Z} \not\preceq  \mathcal{Z}'$}
      { add $(\ell, \mathcal{Z})$ to $\Passed$ \; add
    $\Post(\ell, \mathcal{Z})$ to $\Waiting$ \; } 
} 
\Return{$\Cost$}

\caption{Symbolic algorithm for optimal cost, with inclusion test~$\preceq$}\label{algo:optreach}
\end{algorithm}

Let $\A = (X,L,\ell_0,\Goal,E,\weight)$ be a weighted timed
automaton. The algorithm makes a forward exploration of $\A$ from
$(\ell_0,\mathcal{Z}_0)$ with $\mathcal{Z}_0=(Z_0,\zeta_0)$, where
$Z_0$ is the initial zone defined by $\bigwedge_{x \in X} x=0$ and
$\zeta_0$ is identically~$0$ everywhere. Then, symbolic successors are
iteratively computed, and when the target location is reached, the
minimal cost given by the priced zone is computed (for a priced zone
$\mathcal{Z} = (Z,\zeta)$, we note $\mincost(\mathcal{Z}) = \inf_{v
  \in Z} \zeta(v)$), and compared to the current optimal value
(variable~$\Cost$).  An~inclusion test between priced zones is
performed, which allows to stop the exploration from
$(\ell,\mathcal{Z})$ when $\mathcal{Z} \preceq \mathcal{Z}'$ and
$(\ell,\mathcal{Z}')$ already appears in the set of symbolic states
that have already been explored.  In~\cite{RLS06}, the algorithm uses
the following inclusion test~$\Subset$, which refines the inclusion
test of~\cite{LBB+01}:
inclusion $\mathcal{Z} \Subset \mathcal{Z}'$ holds whenever $Z
\subseteq Z'$ and $\zeta(v) \ge \zeta'(v')$ for every $v \in Z$.
As~shown in~\cite{RLS06}, this algorithm computes the optimal cost
in~$\A$, provided it terminates, and this always happens when the
weights in $\A$ are nonnegative, and when all clocks in~$\A$ are
bounded.

In~the present paper, we~define a refined inclusion
test~$\sqsubseteq$ between priced zones, which will enforce termination
of~\Cref{algo:optreach} even when clocks are not upper-bounded, and,
to some extent, when costs are negative.

\smallskip We now give some definitions which will allow to state the
correctness of the algorithm. Given a timed automaton $\A$, a location
$\ell$ and a priced zone $\mathcal{Z} = (Z, \zeta)$, we say that
$(\ell, \mathcal{Z})$ is \emph{realized} in~$\A$ whenever for every
valuation $v \in Z$, and for every $\epsilon > 0$, there exists a
run~$\varrho$ from the initial state~$(\ell_0,\mathbf{0}_X)$ to
$(\ell, v)$, such that $\zeta(v) \leq \cost(\varrho) \leq \zeta(v) +
\epsilon$.
For a location $\ell$, a priced zone $\mathcal{Z} = (Z,\zeta)$ and a
run $\varrho$ starting in a configuration $s$, we say that $\varrho$
\emph{ends in} $(\ell,\mathcal{Z})$ if $\varrho$~leads from~$s$ to a
configuration $(\ell, v)$ with $v \in Z$ and $\cost(\varrho) \geq
\zeta(v)$.
The post operation $\Post$ on priced zones used
in~\Cref{algo:optreach} is described in~\cite{LBB+01}. Its~computation
is effective (see~\cite{LBB+01}), and is such that (see~\cite{RLS06}):
\begin{itemize}
\item every $(\ell, \mathcal{Z}) \in \Post^*(\ell_0,\mathcal{Z}_0)$ is
  realized in $\A$;
\item for every run $\varrho$ from a configuration $s$ to a
  configuration $s'$, and every mixed move $\tau$ from $s'$, if
  $\varrho$ ends in $(\ell, \mathcal{Z})$, then $\varrho \tau$ ends in
  an element of $\Post(\ell, \mathcal{Z})$.
\item for every run $\varrho$ from $(\ell_0,\mathbf{0}_X)$, there exists $(\ell,
  \mathcal{Z}) \in \Post^*(\ell_0,\mathcal{Z}_0)$ such that $\varrho$
  ends in $(\ell, \mathcal{Z})$ (this is a consequence of the previous property).
\end{itemize}

The purpose of this work is to propose an inclusion test $\sqsubseteq$
such that the following three properties are satisfied:
\begin{enumerate}
\item \textit{(Termination)} \Cref{algo:optreach} with inclusion test
  $\sqsubseteq$ terminates;
\item \textit{(Soudness w.r.t. optimal reachability)}
  \Cref{algo:optreach} with inclusion test $\sqsubseteq$ computes the
  optimal cost for reaching~$\Goal$;
\item \textit{(Effectiveness)}
  There is an algorithm deciding $\sqsubseteq$ on priced zones.
\end{enumerate}
\noindent
We now present our inclusion test, and show its soundness for optimal
reachability. We~then turn to effectiveness (Sect.~\ref{sec:inclusion}),
and then to termination (Sect.~\ref{sec:termination}).

\subsection{The inclusion test}

Our inclusion test is inspired by the inclusion test on (pure) zones proposed
in~\cite{HSW12}.%
\footnote{Contrary to pure reachability, we cannot use the preorder
  $\preceq_{LU}$ (which distinguishes between lower-bounded constraints and
  upper-bounded constraints)~\cite{BBLP05},
  since it does not preserve optimal cost (not
  even optimal time).} We~start by recalling an equivalence relation on
valuations.
We assume a function $M\colon X \mapsto \IN \cup \{ - \infty \}$ such
that $M(x)$ is larger than any constant against which clock~$x$ is compared
to in the (weighted) timed automata under consideration.  Let~$v$ and
$v'$ be two valuations in $\IR_{\ge 0}^X$. Then, $v \equiv_M v'$ iff
for every clock $x \in X$, either $v(x) = v'(x)$, or $v(x) > M(x)$ and
$v'(x) > M(x)$.  We note $[v]_M$ the equivalence class of $v$ under
$\equiv_M$.

\begin{restatable}{lemma}{eqsamefuture}
\label{lemma:equivalence_same_future}
  If $v \equiv_M v'$, then, for any $\ell \in L$,
  $\optcost_\A(\ell,v) = \optcost_\A(\ell,v')$.
\end{restatable}

We now define our inclusion test for two priced zones
$\mathcal{Z} = (Z,\zeta)$ and $\mathcal{Z}' = (Z',\zeta')$; it~is
parameterized by~$M$, which gives upper bounds on clocks:
\[
\mathcal{Z} \sqsubseteq_M \mathcal{Z}'\ \text{iff}\ \forall v \in Z,\
\forall \epsilon>0,\ \exists v' \in Z'\ \text{s.t.}\ v \equiv_M v'\
\text{and}\ \zeta'(v') \le \zeta(v) + \epsilon.
\]

\begin{restatable}{theorem}{theosoundness}
  \label{theo:soundness}
  When using $\sqsubseteq_M$, provided~\Cref{algo:optreach}
  terminates, it is sound w.r.t. optimal reachability (the returned
  cost is the optimal one).
\end{restatable}

\begin{remark}\label{rmk:inclusion}
  Remember that the inclusion test~$\Subset$ of~\cite{RLS06}
  requires $Z \subseteq Z'$ and, for~every $v \in Z$, $\zeta(v) \ge
  \zeta'(v)$. It~is easily seen that $\mathcal{Z} \Subset
  \mathcal{Z}'$ implies $\mathcal{Z} \sqsubseteq_M \mathcal{Z}'$ for
  any~$M$; hence the branches are always stopped earlier in our
  algorithm (which uses $\sqsubseteq_M$) than in the original
  algorithm of~\cite{RLS06} (which uses $\Subset$).
  Moreover, 
  $\Subset$~does not ensure termination of the forward exploration
  when clocks are not bounded: on~the automaton of
  \Cref{fig:non-terminate}~(right), where the optimal time to reach
  the right state is~$10$, the forward algorithm successively computes
  zones $x \le 1 \wedge n \le y-x \le n+1$, for every integer~$n$. Any
  two such zones are always incomparable (for~$\Subset$).
\end{remark}

\section{Effective inclusion check}
\label{sec:inclusion}

In this section we show that we can effectively check the inclusion
test~$\sqsubseteq_M$ of priced zones. For the rest of this section, we~fix two
priced zones $\mathcal{Z} = (Z,\zeta)$ and $\mathcal{Z}' = (Z',\zeta')$, and a
function~$M$. To improve readability, we~write~$\equiv$ and~$\sqsubseteq$ in place
of~$\equiv_M$ and~$\sqsubseteq_M$.
\subsection{Formulation of the optimization problem}
\label{subsec:opt}

We~first express the inclusion of the two priced zones as an
optimization problem.

\begin{restatable}{lemma}{optimisation}
  \label{lemma:inclusion_sup}
  $\mathcal{Z} \sqsubseteq \mathcal{Z}' \iff \sup_{v \in Z}
  \inf_{\substack{v' \in Z'\\v' \equiv v}} \zeta'(v') - \zeta(v) \leq
  0$.
\end{restatable}

Note that $\mathcal{Z} \sqsubseteq \mathcal{Z}'$ already requires some
relation between zones $Z$ and~$Z'$: indeed, for the above inclusion
to hold, it should be the case that for every $v \in Z$, there exists
some $v' \in Z'$ such that $v \equiv v'$.  Interestingly, this
corresponds to the test on (unpriced) zones developed in~\cite{HSW12}
(with~$L=U=M$); this can be done efficiently (in time quadratic in the
number of clocks) as a preliminary test~\cite[Theorem 34]{HSW12}.

\begin{remark}
  The constraint $v \equiv v'$ is not convex, and we have a bi-level
  optimization problem to solve.
  Hence common techniques for convex optimization, such as
  dualization~\cite{BV04}, do~not directly apply to the above
  problem. Still, it~is possible to transform it into finitely many
  so-called \emph{generalized semi-infinite optimization problems}
  (GSIPs)~\cite{RS01} (using~$Z_Y$'s as defined later in this
  section).  As~far as we know, such problems do not have dedicated
  efficient algorithmic solutions. We~thus propose a more direct
  solution, that benefits from the specific structure of our problem
  (see~for instance~\Cref{sec:righty}); it~provides a feasible way to
  solve our optimization problems, hence to decide $\sqsubseteq$ on
  priced zones.
\end{remark}

In order to compute the above optima, we transform our problem into a
finite number of optimization problems that are easier to solve.
Let $Y \subseteq X$.  A~zone~$Z$ is $M$-bounded on~$Y$ if, for every
$v \in Z$, $\{ x \mid v(x) \leq M(x) \} = Y$.  We note $Z_Y$ the
restriction of~$Z$ to its $M$-bounded-on-$Y$ component:
$Z_Y = Z \cap \bigcap_{x \in Y} (x \leq M(x)) \cap \bigcap_{x \notin
  Y} (x > M(x))$.  Note that $Z_Y$ may be empty, and that the
family~$(Z_Y)_{Y\subseteq X}$ forms a partition of~$Z$. We also define
$\mathcal{Z}_Y$ as the priced zone $(Z_Y,\zeta)$. We~define the
natural projection $\pi_Y \colon \IR_{\ge 0}^X \to \IR_{\ge 0}^Y$,
which associates with $v \in \IR_{\ge 0}^X$ the valuation $v' \in
\IR_{\ge 0}^Y$ that coincides with~$v$ on~$Y$.

\begin{figure}[tb]
\colorlet{vert}{green!80!blue}
\colorlet{rouge}{red}
\begin{minipage}[t]{.48\linewidth}
\centering
\begin{tikzpicture}[scale=.4]
\draw[latex'-latex'] (11,0) node[right] {$x$} -| (0,9) node[above] {$y$};
\draw[dashed] (8,0) node[below] {$M(x)$} -- +(0,9);
\draw[dashed] (0,4) node[left] {$M(y)$} -- +(11,0);
\draw[draw=vert,fill=vert!50!white,opacity=.5,line width=1pt] (2,0) -- (2,2) --
  (9,9) -- (9,3) -- (6,0) -- (2,0);
\draw[draw=rouge,fill=rouge!50!white,opacity=.2,line width=1pt] (0,1) -- (6,7) --
  (10,7) -- (10,2) -- (8,0) -- (0,0) -- (0,1);
\draw (6,2.5) node {$Z$};
\draw (.7,.7) node {$Z'$};
\begin{scope}
\path[clip] (3,4) -- (6,7) -- (8,7) -- (8,4) -- (3,4);
\foreach \x in {3,3.4,...,11} {\draw[rouge,opacity=.4] (\x,4) -- +(-4,4);}
\end{scope}
\begin{scope}
\path[clip] (4,4) -- (8,8) -- (8,4) -- (4,4);
\foreach \x in {3.2,3.6,...,11} {\draw[vert,opacity=.4] (\x,4) -- +(0,4);}
\end{scope}
\path[use as bounding box] (0,0);
\draw[line width=1mm,rouge,opacity=.6,cap=round] (3,.04) -- (8,.04);
\draw[loosely dotted] (3,.04) -- (3,4);
\draw[line width=1mm,vert,opacity=.4,cap=round] (4,-.04) -- (8,-.04);
\draw[loosely dotted] (4,.04) -- (4,4);
\path (5.2,1.4) node (a) {$\scriptstyle\pi_{Y}(Z'_{Y})$};
\draw[-latex'] (a.-150) .. controls +(-120:3mm) and +(70:3mm) .. (3.8,.2);
\path (5,-1.4) node (b) {$\scriptstyle\pi_{Y}(Z_{Y})$};
\draw[-latex'] (b) -- (5.5,-.2);
\draw[dotted] (2.3,3.3) -| (8.7,8.7) -| (2.3,3.3);
\path (3.3,6) node (c) {$Z'_Y$};
\draw[-latex'] (c) .. controls +(-70:6mm) and +(140:6mm) .. (4.1,4.6);
\path (5,7.5) node (d) {$Z_Y$};
\draw[-latex'] (d) .. controls +(-70:12mm) and +(140:12mm) .. (6.5,5);
\end{tikzpicture}
\caption{Two-dimensional zones~$Z$ and~$Z'$, and sub-zones $Z_{Y}$
  and $Z'_{Y}$ for $Y=\{x\}$.}
\label{fig-ZY}
\end{minipage}\hfill
\begin{minipage}[t]{.42\linewidth}
\centering
\begin{tikzpicture}[scale=.72]
\draw[dotted] (2.3,3.3) -| (8.7,8.7) -| (2.3,3.3);
\begin{scope}
\path[clip] (2.2,3.2) -| (8.8,8.8) -| (2.2,3.2);
\draw[dashed] (8,0) node[below] {$M(x)$} -- +(0,9);
\draw[dashed] (0,4) node[left] {$M(y)$} -- +(11,0);
\draw[draw=vert,fill=vert!50!white,opacity=.5,line width=1pt] (2,0) -- (2,2) --
  (9,9) -- (9,3) -- (6,0) -- (2,0);
\draw[draw=rouge,fill=rouge!50!white,opacity=.2,line width=1pt] (0,1) -- (6,7) --
  (10,7) -- (10,2) -- (8,0) -- (0,0) -- (0,1);
\begin{scope}
\path[clip] (3,4) -- (6,7) -- (8,7) -- (8,4) -- (3,4);
\foreach \x in {3,3.4,...,11} {\draw[rouge,opacity=.4] (\x,4) -- +(-4,4);}
\end{scope}
\begin{scope}
\path[clip] (4,4) -- (8,8) -- (8,4) -- (4,4);
\foreach \x in {3.2,3.6,...,11} {\draw[vert,opacity=.4] (\x,4) -- +(0,4);}
\end{scope}
\end{scope}
\path[use as bounding box] (5,2.2);
\draw[dotted,-latex',line width=1pt,opacity=.7] (-1.5,7) .. controls +(20:1cm)
  and +(160:1cm) .. (2.1,7); 
\draw[rouge,line width=1mm,opacity=.5,cap=round] (3.1,4.1) -- (5.95,6.95);
\draw[rouge,line width=1mm,opacity=.5,cap=round] (6.1,7) -- (8,7);
\path (4,7) node[text width=2.2cm,align=center] (f) {upper facets of $Z'_Y$ w.r.t.~$y$};
\draw[-latex'] (f) .. controls +(-90:5mm) and +(130:5mm) .. (4,5.1); 
\draw[-latex'] (f.10) .. controls +(20:2mm) and +(160:2mm) .. (6.4,7.1); 
\draw[vert,line width=1mm,opacity=.5,cap=round] (4.05,4.05) -- (8,8);
\path (5.3,8.3) node (g) {upper facet of $Z_Y$ w.r.t.~$y$};
\draw[-latex'] (g.-20) .. controls +(-80:2mm) and +(170:2mm) .. (7.5,7.6); 
\draw[rouge,line width=1mm,opacity=.5,cap=round] (3,3.96) -- (8,3.96);
\draw[vert,line width=1mm,opacity=.5,cap=round] (4.04,4.04) -- (8,4.04);
\path (5.5,2.4) node (h) {lower facets of $Z_Y$ and~$Z'_Y$ w.r.t.~$y$};
\draw[-latex'] (h.160) .. controls +(120:5mm) and +(-90:5mm) .. (3.5,3.8);
\draw[-latex'] (h.50) .. controls +(60:5mm) and +(-90:5mm) .. (6.5,3.8);
\end{tikzpicture}
\caption{Simple facets of $Z_Y$ and $Z'_Y$ w.r.t. clock~$y$.}
\label{fig-facets}
\end{minipage}
\end{figure}

\begin{restatable}{lemma}{timedinclusion}
  \label{lemma:timed_inclusion}
  The following two properties are equivalent:
  \begin{enumerate}[(i)]
  \item for every $v \in Z$, there is $v' \in Z'$ such that $v'
    \equiv v$
  \item for every $Y \subseteq X$, $\pi_Y(Z_Y) \subseteq \pi_Y(Z'_Y)$.
  \end{enumerate}
\end{restatable}

This allows to transform the initial optimization problem into
finitely many optimization problems.

\begin{restatable}{lemma}{decomposition}
  \label{lemma:decompos}
  \quad \(
  \displaystyle
  \sup_{v \in Z} \pinf_{\substack{v' \in Z'\\v' \equiv v}} \zeta'(v')
  - \zeta(v) = \pmax_{Y \subseteq X} \sup_{v \in Z_Y}
  \pinf_{\substack{v' \in Z'_Y\\v' \equiv v}} \zeta'(v') - \zeta(v)
  \).
\end{restatable}

\begin{corollary}
  \label{coro:decompos}
  $\mathcal{Z} \sqsubseteq \mathcal{Z}'$ iff for every $Y \subseteq
  X$, $\mathcal{Z}_Y \sqsubseteq \mathcal{Z}'_Y$
\end{corollary}

In the sequel, we~write 
\[
S(\calZ,\calZ',Y)=\sup_{v \in Z_Y} \inf_{\substack{v' \in Z'_Y \\ v' \equiv v}}
\zeta'(v') - \zeta(v)
\]
\Cref{lemma:inclusion_sup} and~\Cref{coro:decompos} suggest an
algorithm for deciding whether $\calZ \sqsubseteq \calZ'$: enumerate
the subsets~$Y$ of $X$, and prove that $S(\calZ,\calZ',Y) \le
0$. 
We~now show how to solve the latter optimization problem (for~a
fixed~$Y$), and then show how we can drive the choice of~$Y$ so that
not all subsets of~$X$ have to be analyzed.

\subsection{Computing $S(\calZ,\calZ',Y)$}

We show the following main result to compute $S(\calZ,\calZ',Y)$,
which produces a simpler optimization problem, allowing to decide the
inclusion of two priced zones, on parts where cost functions are
lower-bounded.

\begin{restatable}{theorem}{projectionfacettes}
  \label{theo:projections-facets}
  Let $\mathcal{Z}=(Z,\zeta)$ and $\mathcal{Z}'=(Z',\zeta')$ be two
  non-empty priced zones, and let $Y \subseteq X$ be such that
  $\pi_Y(Z_Y) \subseteq \pi_Y(Z'_Y)$ and $\zeta$ and $\zeta'$ are
  lower-bounded on $Z_Y$ and $Z_Y'$ respectively.
  Then we can compute finite sets $\mathcal{K}_Y$ and $\mathcal{K}'_Y$
  of zones over $Y$, and affine functions $\zeta_F$ and $\zeta'_{F'}$
  for every $F \in \mathcal{K}_Y$ and $F' \in \mathcal{K}'_Y$ s.t.:
  \begin{equation}
  S(\calZ,\calZ',Y) = \pmax_{F \in \mathcal{K}_Y} \pmax_{F' \in
    \mathcal{K}'_Y} \sup_{u \in F \cap F'} \zeta'_{F'} (u) - \zeta_F(u).
  \label{eq-thmS}
  \end{equation}
\end{restatable}

The details of the proof of this theorem is given in~\Cref{app:proof}.
The idea behind this result is to first rewrite $S(\calZ,\calZ',Y)$
into:
\[
S(\calZ,\calZ',Y) = \sup_{u \in \pi_Y(Z_Y)}
\Bigl[\Bigl(\inf_{\substack{v' \in Z_Y'\\\pi_Y(v') = u}}
\zeta'(v')\Bigr) - \Bigl(\inf_{\substack{v \in Z_Y\\\pi_Y(v) = u}}
\zeta(v)\Bigr)\Bigr]
\]
which decouples the dependency of $v'$ on $v$.
The algorithm then uses the notion of facets (introduced
in~\cite{LBB+01}), which corresponds to the boundary of the zone
w.r.t. a clock (if~$W$ is the zone, a~facet of~$W$ w.r.t.~$x$ is
$\adherence{W} \cap (x=n)$ or $\adherence{W} \cap (x-y=m)$ whenever $x
\bowtie n$ or $x-y \bowtie m$ is a constraint defining~$W$). Given a
clock $x \in X \setminus Y$, we~consider the facets of~$Z_Y$
w.r.t.~$x$ that minimize, for any $w \in \pi_{X \setminus \{x\}}(Z_Y)$,
the function $v \mapsto \zeta(v)$ when $\pi_{X \setminus \{x\}}(v) = w$.
The~restriction of~$\zeta$ on such a facet is a new affine function,
which we can compute. We~then iterate the process for all clocks in $X
\setminus Y$. We~do the same for~$\zeta'$. This~yields the result
claimed above: sets~$\mathcal{K}_Y$ and~$\mathcal{K'}_Y$ are sets of
projections of facets over~$Y$.

Facets are zones, and so are their projections on~$Y$ and
intersections thereof. Additionally, all functions $\zeta_F$ and
$\zeta'_{F'}$ are affine; hence the supremum in Eq.~\eqref{eq-thmS} is
reached at some vertex~$u_0$ of zone $F \cap F'$, for some facets~$F$
and~$F'$.
By~construction of~$\zeta_F$ and~$\zeta'_{F'}$, we~get
\[
S(\calZ,\calZ',Y) = \inf_{\substack{v' \in Z'_Y \\ \pi(v') = u_0}}
\zeta'(v') - \inf_{\substack{v \in Z_Y\vphantom{Z'_Y} \\ \pi(v) =
    u_0}} \zeta(v)
\]

In particular, $u_0$ has integral coordinates. We end up with the
following result, which will be useful for proving the termination
of~\Cref{algo:optreach}:
\begin{corollary}
  \label{theo:vertices}\label{coro:vertices}
  Let $\mathcal{Z}=(Z,\zeta)$ and $\mathcal{Z}'=(Z',\zeta')$ be two
  non-empty priced zones, and let $Y \subseteq X$ be such that
  $\pi_Y(Z_Y) \subseteq \pi_Y(Z'_Y)$ and $\zeta$ and $\zeta'$ are
  lower-bounded on $Z_Y$ and $Z_Y'$ respectively. Then the following
  holds:
  \[
  S(\calZ,\calZ',Y) = \max_{\substack{u_0 \in \pi_Y(\adherence{Z_Y})
      \\u_0\in\bbN^Y}} \Bigl[
  \min_{\substack{\vphantom{\adherence{Z_Y}}v' \in Z'_Y \\
      v'\equiv u_0}} \zeta'(v') -
  \min_{\substack{\vphantom{\adherence{Z_Y}Z'_Y}v \in Z_Y\\ \vphantom{v'}v\equiv u_0}}
  \zeta(v)\Bigr]
  \]
\end{corollary}

The requirement for lower-bounded priced zones
in~\Cref{theo:projections-facets} is crucial in the proof. 
But the case when this requirement is not met can easily be handled separately, so
that $\sqsubseteq$ can always be effectively decided:

\begin{restatable}{lemma}{relaxation}
  \label{lemma:relaxation}
  Let $\mathcal{Z}=(Z,\zeta)$ and $\mathcal{Z}'=(Z',\zeta')$ be two
  non-empty priced zones.
  \begin{itemize}
  \item If $\zeta$ is not lower-bounded on $Z$ but $\zeta'$ is
    lower-bounded on $Z'$, then $\mathcal{Z} \not\sqsubseteq
    \mathcal{Z}'$.
  \item Let $Y \subseteq X$ such that $\pi_Y(Z_Y) \subseteq \pi_Y(Z_Y')$. If $\zeta'$ is not lower-bounded on
    $Z'_Y$, then $\mathcal{Z}_Y \sqsubseteq \mathcal{Z}'_Y$.
  \end{itemize}
\end{restatable}

\begin{corollary}
  Let $\mathcal{Z}=(Z,\zeta)$ and $\mathcal{Z}'=(Z',\zeta')$ be two
  priced zones. Then we can effectively decide whether $\mathcal{Z}
  \sqsubseteq \mathcal{Z}'$.
\end{corollary}

\subsection{Finding the right $Y$}
\label{sec:righty}

Applying \Cref{lemma:decompos}, the main obstacle to efficiently
decide~$\sqsubseteq_M$ is to find the appropriate $Z_Y$ in which the sought
supremum is reached. Unless good arguments can be found to guide the search
towards the best choice for~$Y$, an exhaustive enumeration of all the $Y$'s
will be required.

\begin{example}
  We consider the zone $Z$ defined by the constraints $x \geq 0$, $y
  \geq 1$, $x \leq y$ and $y \leq x+2$.  We fix $M(x) = 2$ and $M(y) =
  3$.  We then consider $Z' = Z$. The zone~$Z$ is equipped with a constant cost
  function~$\zeta$.
  In~\Cref{fig:examples:middle}, $Z'$ is attached $\zeta'(x,y) = x + y$, and
  the  
  expression of the function 
$f(v) = \inf_{v' \in Z',\ v' \equiv_M v} \zeta'(v')$ 
  is given in each~$Z_Y$, for~$Y\subseteq X$.
  It is then easy to see that the supremum of~$f$ is reached at the
  point 
  $(2,3)$, in the middle of the zone.  In~\Cref{fig:examples:border},
  we~take $\zeta'(x,y) = 2x - y$, and the expression of the function
  $f(v) = \inf_{v' \in Z'.\ v' \equiv_M v} \zeta'(v')$ is given in
  each~$Z_Y$.  The supremum of~$f$ is then reached at the point 
  $(2,2)$, on the border, but not at a corner of the zone.  The latter
  example also shows that $f$ is not continuous on the whole zone~$Z$.
\end{example}

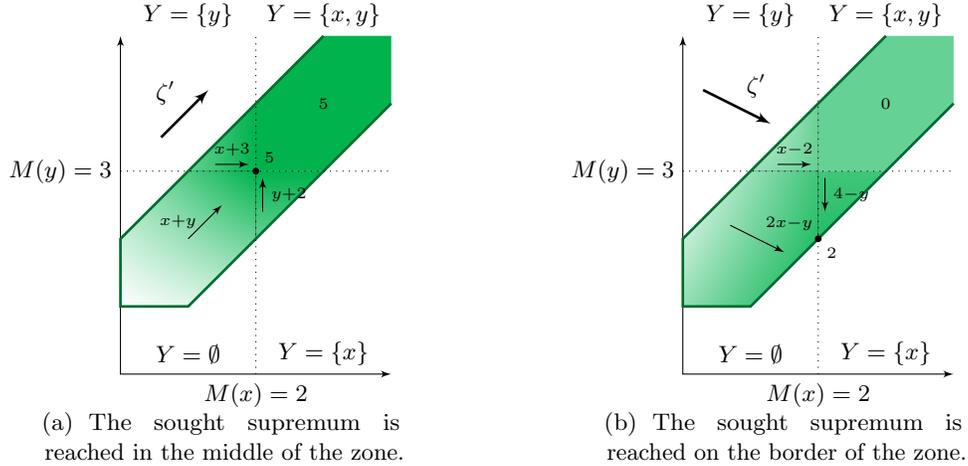
\begin{figure}[tb]
\colorlet{vert}{green!70!blue}
\centering
\subfigure[The sought supremum is reached in the middle of the
zone.\label{fig:examples:middle}]{
\begin{tikzpicture}[scale=.9]
\draw[latex'-latex'] (4,0) -| (0,5);
\path[use as bounding box] (-1,-.4);
\draw[-latex',line width=1pt] (.6,3.5) -- (1.3,4.2) node[midway,above left]
  {$\zeta'$};
\begin{scope}[opacity=1]
\fill[shading=axis,lower left=vert!0!white,upper left=vert!40!white,lower
right=vert!20!white,upper right=vert!100!white] (2,2) -- (1,1) -| (0,2) -- (1,3) -| (2,2);
\fill[shading=axis,lower left=vert!60!white,lower right=vert!60!white,
  upper left=vert,upper right=vert]  (2,2) -- (3,3) -| (2,2);
\fill[shading=axis,lower left=vert!60!white,lower right=vert,
  upper left=vert!60!white,upper right=vert]  (1,3) -- (2,4) |- (1,3);
\fill[vert] (2,4) |- (3,3) -- (4,4) |- (3,5) -- (2,4);
\end{scope}
\draw[vert!60!black,line width=1pt] (4,4) -- (1,1) -| (0,2) -- (3,5);
\draw[dotted] (0,3) node[left] {$M(y)=3$} -- +(4,0);
\draw[dotted] (2,0) node[below] {$M(x)=2$} -- +(0,5);
\draw[-latex'] (1,2) -- +(.5,.5) node[midway,left] {$\scriptstyle x+y$};
\draw[-latex'] (1.4,3.1) -- +(.5,0) node[midway,above] {$\scriptstyle x+3$};
\draw[-latex'] (2.1,2.4) -- +(0,.5) node[midway,right] {$\scriptstyle y+2$};
\draw (3,4) node {$\scriptstyle 5$};
\fill (2,3) circle (.5mm) node[above right] {$\scriptstyle 5$};
\draw (1,.3) node {$Y=\emptyset$};
\draw (3,.3) node {$Y=\{x\}$};
\draw (1,5.3) node {$Y=\{y\}$};
\draw (3,5.3) node {$Y=\{x,y\}$};
\end{tikzpicture}
}
\hfill
\subfigure[The sought supremum is reached on the border of the
zone.\label{fig:examples:border}]{
\begin{tikzpicture}[scale=.9]
\draw[latex'-latex'] (4,0) -| (0,5);
\path[use as bounding box] (-1,-.4);
\draw[-latex',line width=1pt] (.3,4.2) -- (1.3,3.7) node[midway,above right]
  {$\zeta'$};
\begin{scope}
\fill[shading=axis,shading angle=-71,lower left=vert!0!white,
  upper left=vert!60!white,lower right=vert!60!white,upper right=vert]  (2,2) -- (1,1) -| (0,2) -- (1,3) -| (2,2);
\fill[shading=axis,lower left=vert!85!white,lower right=vert!85!white,
  upper left=vert!75!white, upper right=vert!75!white]  (2,2) -- (3,3) -| (2,2);
\fill[shading=axis,lower left=vert!40!white,lower right=vert!60!white,upper right=vert!60!white]  (1,3) -- (2,4) |- (1,3);
\fill[vert!60!white] (2,4) |- (3,3) -- (4,4) |- (3,5) -- (2,4);
\end{scope}
\draw[vert!60!black,line width=1pt] (4,4) -- (1,1) -| (0,2) -- (3,5);
\draw[dotted] (0,3) node[left] {$M(y)=3$} -- +(4,0);
\draw[dotted] (2,0) node[below] {$M(x)=2$} -- +(0,5);
\draw[-latex'] (.7,2.2) -- +(.8,-.4) node[midway,above right] {$\scriptstyle 2x-y$};
\draw[-latex'] (1.4,3.1) -- +(.5,0) node[midway,above] {$\scriptstyle x-2$};
\draw[-latex'] (2.1,2.9) -- +(0,-.5) node[midway,right] {$\scriptstyle 4-y$};
\draw (3,4) node {$\scriptstyle 0$};
\fill (2,2) circle (.5mm) node[below right] {$\scriptstyle 2$};
\draw (1,.3) node {$Y=\emptyset$};
\draw (3,.3) node {$Y=\{x\}$};
\draw (1,5.3) node {$Y=\{y\}$};
\draw (3,5.3) node {$Y=\{x,y\}$};
\end{tikzpicture}
}
\caption{The supremum may lie in the middle of zones or facets}
\label{fig:examples}
\end{figure}

Nevertheless, in many cases, we will be able to guide the search of the $Z_Y$
where the sought optimal is to be found.
The~following development
focuses on the zone, not on the cost function. Given a zone~$Z$, we define a
preorder~$\preceq$ on the clocks, such that if $Z_Y \neq \emptyset$, then
$Y$~is downward-closed for~$\preceq$. In~other words, whenever $x \preceq y$,
$y \in Y$ and $Z_Y \neq \emptyset$, then $x \in Y$. The knowledge of $\preceq$
can be a precious help to guide the enumeration of non-empty~$Z_Y$'s.
Indeed, if $Z_Y \neq \emptyset$, $Y$~is downward-closed for~$\preceq$, and
candidates for~$Y$ are thus found by enumerating the antichains of~$\preceq$.
In particular, if $\preceq$ is total, then there are at most $|X|+1$ sets $Y$
such that $Z_Y \neq \emptyset$.

To be concrete, 
let $X_{\le M}$ and~$X_{>M}$ be the (disjoint) sets of clocks~$x$ such that $Z
\subseteq (x \le M(x))$ and ${Z \subseteq (x>M(x))}$, respectively.
We~define the relation~$\preceq_Z$ as the
least
relation satisfying the following conditions:
\begin{itemize}
\item for each $x \in X_{\le M}$, for each $y \in X$, $x \preceq_Z y$;
\item for each $y \in X_{>M}$, for each $x \in X$, $x \preceq_Z y$;
\item for each $y \in X_{>M}$, for each $x \in X \setminus X_{>M}$, $y
  \not\preceq_Z x$;
\item for all $x,y \in X \setminus X_{>M}$, $Z \subseteq (x-y \le
  M(x)-M(y))$ implies $x \preceq_Z y$.
\end{itemize}
It is not difficult to show that $\preceq_Z$ is a preorder such that:
$y \in X_{\le M}$ and $x \preceq_Z y$ implies $x \in X_{\le M}$, and
$x \in X_{>M}$ and $x \preceq_Z y$ implies $y \in X_{>M}$.

\begin{restatable}{lemma}{Ydownwardclosed}
  Let $Y \subseteq X$ such that $Z_Y \ne \emptyset$. Then $Y$ is
  downward-closed for $\preceq_Z$.
\end{restatable}

The preorder $\preceq_Z$ can be computed in polynomial time, since it
only requires to check emptiness of zones, which can be done in time
polynomial in $|X|$ (cubic in $|X|$ with DBMs for instance).

We recall that, if $Z$ is a zone generated in a timed automaton where
only resets of clocks to $0$ are allowed, for any pair of clocks
$x,y$, it~cannot be the case that $Z$ crosses the diagonal hyperplane of
equation $x = y$. 

\begin{proposition}
  \label{prop:total_samebound}
  If $Z$ is generated by a timed automaton, and all clocks have the
  same bound $M$, then $\preceq_Z$ is total.
\end{proposition}

\begin{proof}
  Let $x$ and $y$ be two clocks. Since $Z$ is generated by a timed automaton,
  it~is contained either in the half-space of equation $[x \leq y]$, or in the
  one of equation $[x \geq y]$. By definition of $\preceq_Z$, and since $M(x)
  = M(y)$, the former entails $x \preceq_Z y$, and the latter $y \preceq_Z x$.
  Any two clocks are thus always comparable, and $\preceq_Z$ is therefore
  total. \qed
\end{proof}

Under the assumptions of~\Cref{prop:total_samebound}, there are
polynomially many subsets $Y\subseteq X$ to~try.
  Note that these assumptions are easily realized by taking
  $\widetilde{M} = \max_{x \in X} M(x)$ as the unique maximal constant for
  all the clocks.
  Formally, $\sqsubseteq_{\widetilde{M}}$~is~an under-approximation of
  the exact version of $\sqsubseteq_M$.  This approximation does not
  hinder correctness, and illustrate the trade-off between the
  complexity of the inclusion procedure and the number of priced zones
  that will be explored.

\section{Termination of the computation}
\label{sec:termination}

In this section we prove termination of our algorithm, by 
exhibiting an appropriate well-quasi-order.
We fix a timed automaton~$\mathcal{A}$ and a maximal-constant function~$M$
(for~every clock $x \in X$, the~integer~$M(x)$ is larger than any constant
with which clock~$x$ is compared in~$\mathcal{A}$).

\begin{restatable}{proposition}{preorder}
  $\sqsubseteq$ is a preorder (or quasi-ordering).
\end{restatable}

We now consider the ``converse''
preorder~$\sqsupseteq$, defined over priced zones by $\mathcal{Z}'
\sqsupseteq \mathcal{Z}$ iff $\mathcal{Z} \sqsubseteq
\mathcal{Z}'$.  We~show that $\sqsupseteq$ is a \emph{well
  quasi-ordering~(wqo)}.  Thus the relation~$\sqsupseteq$ has no
infinite antichain, which 
entails
termination of
Algorithm~\ref{algo:optreach}. 

We now gather the results to exhibit a sufficient condition for~$\sqsupseteq$
to be a~wqo.

\begin{restatable}{theorem}{wqobis}
  For every $\mu \in \mathbb{Z}$, $\sqsupseteq$ is a well-quasi-order
  on (non-empty) priced zones whose cost functions are either not
  lower-bounded, or lower-bounded by $\mu$.
\end{restatable}

\begin{corollary}
  \Cref{algo:optreach} terminates on weighted timed automata, which
  generate priced zones with a uniform lower bound on the cost
  functions,
\end{corollary}

We can argue (see Appendix~\ref{app-mu}) that infinite antichains for
$\sqsupseteq$ generated by a forward exploration of $\A$ actually
corresponds to infinite paths in $\A$ with cost $-\infty$. While this
condition can be decided (using the corner-point abstraction
of~\cite{BBL08}), we do not want to check this as a preliminary step,
since this is as complex as computing the optimal cost. Furthermore,
symbolically, this would amount to finding a cycle of symbolic states
which is both $\omega$-iterable~\cite{JR11,DH+14} and cost-divergent;
this is a non-trivial problem. We can nevertheless give simple
syntactic conditions for the condition to hold: this is the case of
weighted timed automata with non-negative weights (this is the class
considered in~\cite{LBB+01,RLS06}); let $T_\ell$ be the minimum
(resp. maximum) delay that can be delayed in $\ell$ if location $\ell$
has positive (resp. negative) cost: if along any cycle of the weighted
timed automaton, the sum of the discrete weights and of each $T_\ell
. \weight(\ell)$ is nonnegative, then the above condition will be
satisfied; this last condition encompasses all the acyclic weighted
timed automata, like all scheduling problems~\cite{BLR05b}.

\section{Experimental Results}
\label{sec:exp}

We have implemented a prototype, \tiamo{}, to test our new inclusion
test.  It is based on the DBM library of Uppaal
(in~C++),\footnote{\url{http://people.cs.aau.dk/~adavid/UDBM/}} which
features the inclusion test of~\cite{RLS06}.  We added our inclusion
test (also in~C++).  This core is then wrapped in OCaml code, in which
the main algorithm is written.
Note that \tiamo{} relies on the representation of priced zones from
the DBM library of Uppaal, which assumes that all cost functions are
always non-negative.  For this reason, \tiamo{} does currently not
support models with negative costs.

\tiamo{} is able to prune the state space using the best cost so~far.
Concretely, it~would not explore states whose cost exceeds the current
optimal cost. This can dramatically reduce the state space to explore,
but is sound only when all costs in the model are non-negative.
On such models, the user can provide a hint, a known cost to \tiamo{}
(obtained for example by a reachability analysis, or by other independent
techniques) to be used to prune the model. 
Moreover, \tiamo{} reports, during the computation, the best known
cost so~far. Such values are upper bounds on the sought optimum, and
may be interesting to get during long computations.

A direct comparison between \tiamo{} and
Uppaal\footnote{\url{http://www.uppaal.org/}}
(or~Uppaal-CORA\footnote{\url{http://people.cs.aau.dk/~adavid/cora/}})
is difficult: the source code of Uppaal (and Uppaal-CORA) is not open,
and it is often hard to know what is precisely implemented. For
instance, on the unbounded automaton of~\Cref{fig:non-terminate},
the~algorithm described in~\cite{LBB+01,RLS06} does not
terminate. Depending on the way it is queried (asking for the fastest
trace, or with an \texttt{inf} query), Uppaal terminates or runs
forever on this model.

In order to measure the impact of the inclusion test on the algorithm, we
decided to compare the performance of \tiamo{} running one or the other
inclusion test ($\Subset$~or~$\sqsubseteq$). Our~primary concern is to compare
the number of (symbolic) states explored, and the number of inclusion tests
performed.

We run our experiments with and without pruning activated. Deactivated pruning
allows to measure the impact of the choice of the inclusion test itself. It~is
also more representative of the behavior that can be expected on models with
negative costs, for which pruning is not sound.

\begin{table}[tb]
\centering
\scriptsize
\newcolumntype{P}{>{\raggedleft}p{10mm}}
\newcolumntype{Q}{>{\raggedleft}p{12mm}}
\def\arraystretch{1.2}
\begin{tabular}{|l|c|c|c|c|c|c|c|c|c|c|}
\cline{2-11}   \multicolumn{1}{c|}{}              & \multicolumn{4}{c|}{ALS}
& \multicolumn{2}{c|}{ETS}
& \multicolumn{4}{c|}{\Cref{fig:non-terminate} (right)} \\
\cline{2-11}   \multicolumn{1}{c|}{}              & \multicolumn{2}{c|}{SBFS+P}   & \multicolumn{2}{c|}{SBFS}        & \multicolumn{2}{c|}{SBFS+P}       & \multicolumn{2}{c|}{SBFS+P} & \multicolumn{2}{c|}{SBFS} \\
\cline{2-11}   \multicolumn{1}{c|}{}              & $\sqsubseteq$   & $\Subset$     & $\sqsubseteq$ & $\Subset$         & $\sqsubseteq$ & $\Subset$           & $\sqsubseteq$   & $\Subset$   & $\sqsubseteq$   & $\Subset$ \\
\hline
$\#$ added to $\Waiting$    & 12,016          & 26,549        & 73,487        & 908,838           & 107           & 664                 & 14              & NT          & 14              & NT \\
$\#$ added to $\Passed$     & 4,914           & 9,687         & 63,914        & 908,612           & 84            & 606                 & 13              & NT          & 14              & NT \\
max. $\#$ stored            & 9,705           & 21,407        & 40,297        & 599,721           & 83            & 590                 & 14              & NT          & 14              & NT \\
\hline
$\#$ tests                  & 351,198         & 1,598,135     & 14,508,440
& 
  $>9\cdot 10^9$
& 174           & 17,684              & 135             & NT         & 135             & NT \\
$\#$ successful tests             & 11,796          & 27,885        & 354,126       & 3,627,958         & 66            & 455                 & 3               & NT           & 3               & NT \\
\hline
time                        & $< 2$ s         & $< 2$ s       & 12 s          & 41 min            & $< 1$ s       & $< 1$ s             & $< 1$ s         & NT     & $< 1$ s         & NT \\
\hline
\end{tabular}
\smallskip

\caption{Experimental results\label{table:exp}}
\end{table}

\paragraph{The models.}
We briefly describe the models used in our experiments. The first two are
case studies described on the web page of Uppaal-CORA.

The Aircraft Landing System (ALS) problem has been described
in~\Cref{fig-exintro}: it~consists in scheduling landings of aircrafts
arriving to an airport with several runways.
In our model, we had two runways and 7 aircrafts.
The model has $5$ clocks (one global clock, plus two per runway) and $48,000$ discrete states.

In the Energy-optimal Task-graph Scheduling (ETS) problem, several processors
having different speeds and powers are to be used to perform interdependent
tasks. The aim is to optimize energy consumption for performing the given set
of tasks within a certain delay. The model we used for our experiments is the
one described in~\cite[Example~3]{BFLM11}.
It has $2$ clocks (one per CPU) and $55$ discrete states.

Finally, we also ran \tiamo\ on the model~\Cref{fig:non-terminate}, to
illustrate that $\sqsubseteq$ handles unbounded models.
This model has two clocks and two discrete states.

\paragraph{Exploration strategies.}
\tiamo{} implements several strategies to explore the symbolic state
space.  We retain here only the one called SBFS, a modification of BFS
based on the observation that, if $s$ subsumes~$s'$, the successors
of~$s'$ are subsumed by successors of~$s$. Successors of~$s$ are thus
explored first, until all successors of~$s'$ in the $\Waiting$ list
are subsumed. This is a very naive implementation of a strategy
proposed in~\cite{HT15}. The strategy has two variants, depending on
whether pruning is activated~(SBFS+P) or~not~(SBFS).  For the ETS
problem, both strategies yield very similar results, so~we chose to
only present SBFS+P.

\paragraph{Experimental results.}
The results are summed up in~\Cref{table:exp}. For each model, and for
different combinations of inclusion test and exploration strategy, we indicate
the number of symbolic states explored, as~well as the number of tests
(successful or~not) that have been performed. 
We also indicate the maximal size of the list $\Passed$; although not detailed in~\Cref{algo:optreach},
the tool ensures that $\Passed$ remains an antichain.
This minimizes the number of inclusion tests.
When a new element is added to the $\Passed$ list, all elements of $\Passed$ subsumed by
the new one are removed, so that the size of $\Passed$ does not necessarily increase.

The mention ``NT'' means that the computation does not terminate.
We~observe that $\sqsubseteq$ always explores fewer states
than~$\Subset$, for any given exploration strategy. Though this was
expected (recall~\Cref{rmk:inclusion}), we believe the reduction is
impressive.  It is significant even for small models
(such~as~ETS). The~case of ALS with no pruning shows that the higher
complexity of~$\sqsubseteq$ can be largely compensated by the
reduction in the size of the state space to explore.
On the model of~\Cref{fig:non-terminate}, our inclusion~$\sqsubseteq$ ensures
termination, while $\Subset$ does~not.

\section{Conclusion}

In this paper we have built over a symbolic approach to the
computation of optimal cost in weighted timed
automata~\cite{LBB+01,RLS06}, by proposing an inclusion test between
priced zones. Using that inclusion test, the forward symbolic
exploration terminates and computes the optimal cost for all weighted
timed automata, regardless whether clocks are bounded or not. The idea
of this approach is based on recent works on pure timed
automata~\cite{HSW12}, where a clever inclusion test ``replaces'' any
abstraction computation during the exploration.

We will pursue our work with extensive experimentations using our tool
\tiamo{}. We will also look for more dedicated methods for specific
application domains, like planning problems.

\putbib[biblio,extra]
\end{bibunit}

\newpage
\appendix
\makeatletter
\@addtoreset{theorem}{section}
\makeatother
\setcounter{theorem}{0}
\def\thetheorem{\Alph{section}.\arabic{theorem}}
\def\thelemma{\Alph{section}.\arabic{lemma}}
\def\thecorollary{\Alph{section}.\arabic{corollary}}
\def\theproposition{\Alph{section}.\arabic{proposition}}
\def\thedefinition{\Alph{section}.\arabic{definition}}
\begin{bibunit}[alpha]
\section{A counter-example to the algorithm of~\cite{ARF14}}
\label{app:australiens}

\Cref{algo:australiens} is the algorithm proposed in~\cite{ARF14} to compute
optimal-time reachability in timed automata. Briefly, the approach considers a
global clock, denoted by $x_i$, to measure the duration of the current run.
The algorithm explores the state space in a forward manner.
Newly encountered zones are compared against those already explored using the
abstract inclusion test described in~\cite{HSW12}, with the particularity that
the global clock~$x_i$ is \emph{not} considered in this test. The~clock~$x_i$
keeps track of the time elapsed so~far, and plays no role in the transition
relation, and this is why it does not need to be used in the comparison of
zones. 

The abstract test between $Z$ and $Z'$ is noted $Z \subseteq
closure_{\alpha}(Z')$. Zones are represented as sets of clock constraints,
$\textit{UC}$~denotes the set of constraints involving the global clock~$x_i$.
Thus, $Z \setminus \textit{UC}$ denotes the zone~$Z$ where constraints on~$x_i$
are ignored (though this is not formally defined in~\cite{ARF14}).

\SetKwData{lowerbound}{lowerBound}
\SetKwData{Cost}{BCET}

\begin{algorithm}[ht]

\DontPrintSemicolon

$\Cost \gets \infty$ \;
$\Passed \gets \emptyset$ \;
$\Waiting \gets \{ (\ell_0, Z_0) \}$ \;
\While{$\Waiting \neq \emptyset$}{
  select $(\ell, Z)$ from $\Waiting$ \;
  \tcp{In \cite{ARF14}, $\Goal$ is the set of final states, \ie{} states with no successors}
  \If{$(\ell, Z) \in \Goal$ and $\lowerbound(Z, x_i) < \Cost$}{
    $\Cost \gets \lowerbound(Z, x_i)$ \;
  } 
  add $(\ell, Z)$ to $\Passed$ \;
  \ForAll{$(\ell',Z') \in \Post(\ell, Z)$}{
    \tcp{check if lower bound of $x_i$ in the new zone is less than the best known BCET}
    \If{$\lowerbound(Z', x_i) < \Cost$}{
      \If{$(Z' \setminus \textit{UC}) \not\subseteq closure_{\alpha}(Z'' \setminus \textit{UC})$ for all $(\ell', Z'') \in \Passed$}{
        add $(\ell', Z')$ to $\Waiting$ \;
      } 
    } 
  } 
} 
\Return{$\Cost$}

\caption{Algorithm for optimal time reachability
  from~\cite{ARF14}}\label{algo:australiens} 
\end{algorithm}

We claim that the algorithm fails to find the BCET of the automaton shown
in~\Cref{fig:contrexemple}. We~show this by describing the run
of~\Cref{algo:australiens} on this automaton.

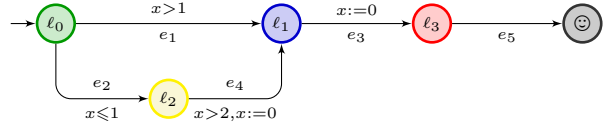
\begin{figure}[ht]
  \centering
  \begin{tikzpicture}
      \everymath{\scriptstyle}
      
      \draw (0,0) node [rond, colvert,inner sep=0pt] (l0) {$\ell_0$};
      
      \draw (3,0) node [rond,colbleu,inner sep=0pt] (l1) {$\ell_1$};

      \draw (1.5,-1) node [rond,coljaune,inner sep=0pt] (l2) {$\ell_2$};

      \draw (5,0) node [rond,colrouge,inner sep=0pt] (l3) {$\ell_3$};

      \draw (7,0) node [inner sep=0pt,rond,colgris] (lf)
      {$\smiley$};

      \draw [-latex'] (l0) -- (l1) node [midway,above]
        {$x > 1$} node [midway,below] {$e_1$};

      \draw [-latex',rounded corners=2mm] (l0) |- (l2) node
            [pos=.75,above] 
      {$e_2$} node [pos=.75,below] {$x \leq 1$};

      \draw [-latex'] (l1) -- (l3) node [midway,below,sloped]
      {$e_3$} node [midway,sloped,above] {$x:=0$};
      
      \draw [-latex'] (l3) -- (lf) node [midway,below,sloped]
      {$e_5$} node [midway,sloped,above] {};

      \draw [-latex',rounded corners=2mm] (l2) -| (l1) node
            [pos=.25,below] {$x > 2, x := 0$} 
      node [pos=.25,above] {$e_4$};

      \draw [latex'-] (l0) -- +(-.6,0);
    \end{tikzpicture}
  \caption{Counter-example to~\Cref{algo:australiens}}
  \label{fig:contrexemple}
\end{figure}

Assume that the state space is explored depth-first.
The initial state is $s_0 = (\ell_0, (x = x_i = 0))$. Initially, $\Cost =
\infty$, $\Passed = \emptyset$ and $\Waiting = \{s_0\}$. State~$s_0$ has two
successors $s_1 = (\ell_1, (x = x_i \wedge x > 1))$ and $s_2 = (\ell_2, (x =
x_i \wedge x \leq 1))$. It~holds $\lowerbound(s1,x_i) = 1 < \Cost$ and
$\lowerbound(s2,x_i) = 2 < \Cost$, so both states are added to~$\Waiting$,
while $s_0$ is added to~$\Passed$.

If $s_2$ is the next state being picked out of~$\Waiting$, it~is easy to see
that its successors $s_3 = (\ell_1, (x_i > 2 \wedge x = 0))$ and $s_4 =
(\ell_3, (x_i > 2 \wedge x = 0))$ are added to~$\Passed$. At~the~end of this
branch, state~$s_5 = (\ell_f, (x_i > x+2))$ is final, and its lower value
for~$x_i$ is used to update the value of~$\Cost$. We now have $\Cost = 2$.

$\Waiting$ now contains a single state~$s_1$. Its single successor is $s_6 =
(\ell_3, (x_i > 1 \wedge x = 0))$. Its lower value for $x_i$ is~$1$, which is
smaller than the current value for~$\Cost$. But $((x_i > 1 \wedge x = 0)
\setminus \textit{UC}) \subseteq closure_{\alpha}((x_i > 2 \wedge x = 0)
\setminus \textit{UC})$. Indeed, ``$\textit{UC}$ represents the set of
constraints involving the extra clock $x_i$'', so $((x_i > 1 \wedge x = 0)
\setminus \textit{UC}) = (x=0) = ((x_i > 2 \wedge x = 0) \setminus
\textit{UC})$. Since $Z \subseteq closure_{\alpha}(Z)$ for every $Z$,
we~obtain the inclusion above. The~paper further states that ``it~is necessary
to check inclusion between zones with respect only to the automaton clocks'',
which confirms our understanding of the operation~``$\setminus \textit{UC}$''.

Therefore, the state $s_6$ is \emph{not} added to~$\Waiting$, which is now
empty, and the algorithm terminates, returning the value $\Cost = 2$.

It is clear however that the BCET of the given automaton is~$1$, a~value that
would have been discovered by the algorithm if~$s_1$ had been picked out of
$\Waiting$ before~$s_2$. However the paper does not mention any assumptions
regarding the order of exploration that would invalidate our counter-example.
We contacted the authors of~\cite{ARF14}, asking to get access to the prototype
implementation mentioned in the paper, but did not get an answer.

\section{Details for~\Cref{sec:algo}}

\subsection{Proof of~\Cref{lemma:equivalence_same_future}}

\eqsamefuture*

\begin{proof}
  Pick a run $\varrho$ starting at $s=(\ell,v)$, and assume $\varrho =
  s \xrightarrow{t_1,e_1} s_1 \dots \xrightarrow{t_p,e_p} s_p$. Then
  $\varrho' = s' \xrightarrow{t_1,e_1} s'_1 \dots
  \xrightarrow{t_p,e_p} s'_p$ is also a run of $\A$ from $s'$ such
  that $s_p \equiv_M s'_p$, and $\cost(\varrho) = \cost(\varrho')$.
  This can be easily shown by induction on $p$. This reinforces the
  property stated in~\cite{BBLP05} saying that $\equiv_M$ is a strong
  timed bisimulation. 
\qed
\end{proof}

\subsection{Proof of~\Cref{theo:soundness} (soundness of~\Cref{algo:optreach})}
\label{app:soundness}
\SetKwData{Cost}{\textsc{Cost}}

\theosoundness*

\begin{proof}
  We assume~\Cref{algo:optreach} terminates.
  Remark that all elements added to the sets $\Waiting$ or $\Passed$
  are obtained through the \Post\ operation from elements of
  $\Waiting$.  Initially, $\Waiting$ contains the initial priced zone
  $(\ell_0,\mathcal{Z}_0)$, hence $\Waiting$ only contains elements
  from $\Post^*(\ell_0,\mathcal{Z}_0)$.  Such elements are necessarily
  realized in $\A$, hence the computed $\Cost$ is not smaller than the
  actual optimal cost $C = \optcost_\A(s_0)$: $\Cost \ge C$.

  \begin{itemize}
  \item Assume that $C$ is finite. Let $\epsilon>0$ and
    $\rho=(\ell_0,v_0) \xrightarrow{t_1,e_1} (\ell_1,v_1) \dots
    \xrightarrow{t_p,e_p} (\ell_p,v_p)$ be an $\epsilon/3$-optimal run
    ($\ell_p$ is a target location):
    $\cost(\rho) \le C+\epsilon/3$.  Unless it is already the case for
    $\rho$, we will build a run $\widehat\rho$ such that all symbolic
    states along $\widehat\rho$ are added to $\Passed$, and
    $\cost(\widehat\rho) \le C+\epsilon$. In particular, if
    $(\widehat\ell,\widehat{\mathcal{Z}})$ is the last symbolic state
    obtained along $\widehat\rho$, and $\widehat{v}$ is the last
    valuation along $\widehat\rho$, $\mincost(\widehat{\mathcal{Z}})
    \le \widehat\zeta(\widehat{v}) \le \cost(\widehat\rho) \le
    C+\epsilon$. Hence the computed cost $\Cost$ is bounded by
    $C+\epsilon$. As $\epsilon>0$ was arbitrary, we conclude that the
    algorithm really computes the optimal cost.

    Assume that not all symbolic states along $\rho$ are added to
    $\Passed$, and write $\rho_1$ for the longest prefix of $\rho$
    along which all symbolic states are added to $\Passed$. Write
    $\tau$ for the next move along $\rho$, and $\rho_2$ for the suffix
    after $\rho_1 \tau$. Let $(\ell,v)$ be the configuration at the
    end of $\rho_1 \tau$, and $(\ell,\mathcal{Z})$ be the
    corresponding symbolic state (which we know, by hypothesis, has
    not been added to $\Passed$). Hence, at the time when
    $(\ell,\mathcal{Z})$ is handled, there exists some
    $(\ell,\mathcal{Z}')$ in $\Passed$, such that $\mathcal{Z}
    \sqsubseteq_M \mathcal{Z}'$.  Thus, there is some valuation $v' \in Z'$
    such that $v' \equiv_M v$ and $\zeta'(v') \le
    \zeta(v)+\epsilon/3$.  Since $(\ell,\mathcal{Z}') \in \Passed$,
    there is a run $\rho'$ from $(\ell_0,v_0)$ to $(\ell,v')$ such
    that (a)~all symbolic states along $\rho'$ have been added to
    $\Passed$, and (b)~$\cost(\rho') \le \zeta'(v') +\epsilon/3 \le
    \zeta(v) + 2\epsilon/3$. Now, since $v' \equiv_M v$, thanks to the
    proof of~\Cref{lemma:equivalence_same_future}, there is a run
    $\rho''$ from $(\ell,v')$ to the target state, with the same
    length as $\rho_2$, such that $\cost(\rho'') = \cost(\rho_2)$. The
    new run $\widetilde\rho = \rho' \cdot \rho''$ reaches the target
    location and its cost is: $\cost(\widetilde\rho) = \cost(\rho') +
    \cost(\rho'') \le \zeta(v) + 2 \epsilon/3 + \cost(\rho_2) \le
    \cost(\rho_1 \tau) + 2 \epsilon/3 + \cost(\rho_2) = \cost(\rho) +
    2\epsilon/3 \le C+\epsilon$.

    We repeat the argument on $\rho_2$, which is smaller than the
    initial run $\rho$, and inductively we are able to build the
    expected run $\widehat\rho$ that we described earlier. 
  \item Assume that $C=\infty$. This means that the target location is
    not reachable, and the algorithm never updates variable $\Cost$;
    its value is therefore $\infty$ as well.
  \item Assume that $C=-\infty$. A reasoning similar to the first case
    can be done, with a slight adaptation: we fix some arbitrary bound
    $K$ and $\epsilon>0$, and we pick the run $\rho$ such that
    $\cost(\rho) \le K$. The rest is identical: we will build a run
    $\widehat\rho$ such that all symbolic states along $\widehat\rho$
    are added to $\Passed$, and $\cost(\widehat\rho) \le \cost(\rho)
    +\epsilon$. \qed
  \end{itemize}
  
\end{proof}

\section{Details for~\Cref{sec:inclusion}}

\subsection{Formulation of the optimization problem}

\optimisation*

\begin{proof}
  \begin{description}
  \item[$(\Rightarrow)$] By definition of $\sqsubseteq$, for every
    $\epsilon > 0$ and $v \in Z$, there is some $v' \in Z'$ such that
    $v \equiv v'$ and $\zeta'(v') \leq \zeta(v) + \epsilon$.  The
    existence of such a $v'$ guarantees that the infimum is not taken
    over an empty set, and hence is not $+\infty$.  Thus, for every $v
    \in Z$, $\inf_{\substack{v' \in Z'\\v' \equiv v}} \zeta'(v') -
    \zeta(v) \leq 0$, from which the result is straightforward.  Note
    that if $Z$ is empty, then the supremum is $-\infty$, and the
    result still holds.
  \item[$(\Leftarrow)$] Let $v \in Z$ and $\epsilon > 0$.
    $\inf_{\substack{v' \in Z'\\v' \equiv v}} \zeta'(v') - \zeta(v)
    \leq 0$, which guarantees the existence of $v' \in Z'$ such that
    $v' \equiv v$ (otherwise this infimum, and the above supremum
    also, would be $+\infty$).  Furthermore, $v'$ can be chosen so
    that $\zeta'(v') - \zeta(v) \leq \epsilon$. \qed
  \end{description}
\end{proof}

\begin{lemma}\label{lem-equivpi}
  Let $Y\subseteq X$, $v\in Z_Y$ and $v'\in Z'$. Then $v\equiv v'$ if, and
  only~if, $v'\in Z'_Y$ and $\pi_Y (v)=\pi_Y(v')$.
\end{lemma}

\begin{proof}
  If~$v\equiv v'$, then for all~$x\in X$, we~have $v(x)=v'(x)$ or both $v(x)$
  and $v'(x)$ are larger than~$M$. As~$v\in Z_Y$, it~holds $v(x)\leq M$ for~$x\in
  Y$ and $v(x)>M$ for~$x\notin Y$. Hence $v'(x)=v(x)$ for~$x\in Y$ and
  $v'(x)>M$ for $x\notin Y$. It~follows $v'\in Z'_Y$ and $\pi_Y
  (v)=\pi_Y(v')$.

  Conversely, if $v'\in Z'_Y$ and $\pi_Y (v)=\pi_Y(v')$, then $v'(x)>M$
  if~$x\notin Y$, and $v(x)=v'(x)$ for~$x\in Y$. Hence $v\equiv v'$.\qed
\end{proof}

Lemma~\ref{lem-equivpi} entails the following two lemmas:

\timedinclusion*

\begin{proof} 
  Assume \textit{(i)}~holds, and pick $Y\subseteq X$. If~$Z_Y$ is empty, the
  inclusion $\pi_Y(Z_Y) \subseteq \pi_Y(Z'_Y)$ is trivial. Otherwise,
  pick~$v\in Z_Y$. Applying~\textit{(i)}, there exists~$v'\in Z'$ s.t.
  $v\equiv v'$. From~Lemma~\ref{lem-equivpi}, $v'\in Z'_Y$ and
  $\pi_Y(v)=\pi_Y(v')$. Hence $\pi_Y(Z_Y) \subseteq \pi_Y(Z'_Y)$.

  Conversely, assuming~\textit{(ii)}, pick~$v\in Z$. Then for some
  (unique)~$Y$, we~have $v\in Z_Y$, so that there exists $v'\in Z'_Y$ for
  which $\pi_Y(v)=\pi_Y(v')$. Applying Lemma~\ref{lem-equivpi}, it~also holds
  $v'\equiv v$, which concludes the proof.\qed
\end{proof}

\decomposition*

\begin{proof}
  It~suffices to notice that when $v\in Z_Y$, the sets $\{v\in Z\mid v\equiv
  v'\}$ and $\{v\in Z_Y \mid v\equiv v'\}$ coincide.\qed
\end{proof}

\subsection{Proof of~\Cref{theo:projections-facets}}
\label{app:proof}

\projectionfacettes*

For the rest of this section, we fix $Y \subseteq X$ such that $Z_Y
\neq \emptyset$ (otherwise $S(\calZ,\calZ',Y)=-\infty$, and we can
take $\mathcal{K} = \mathcal{K}' = \emptyset$).  We~assume that
$\pi_Y(Z_Y) \subseteq \pi_Y(Z_Y')$ (otherwise by
Lemma~\ref{lemma:timed_inclusion}, for some~${v\in Z_Y}$, the~set
$\{v'\in Z'_Y\mid v'\equiv v\}$ would be empty, and
$S(\calZ,\calZ',Y)=+\infty$).
In~this section, we~furthermore require that $\zeta$ and~$\zeta'$ are
bounded from below over $Z_Y$ and~$Z'_Y$, respectively.  Under these
assumptions, we~describe an algorithm to compute~$S(\calZ,\calZ',Y)$.

First remark that, for $v \in Z_Y$ and $v' \in Z_Y'$, it~holds $v
\equiv v'$ if, and only~if, $\pi_Y(v) = \pi_Y(v')$. In~this
subsection, we frequently use this characterization, which can be
extended by continuity to the boundaries of $Z_Y$ and $Z'_Y$.

\paragraph{Facets and decomposition using facets.}
We recall the notion of facets, originally defined in~\cite{LBB+01}. Whenever
$x \bowtie n$ (resp. $x-y \bowtie m$) is a tight constraint defining a
\emph{closed} zone~$W$, the~strengthened zone $W \wedge (x = n)$
(resp.~$W\wedge (x-y = m)$) is called a \emph{simple facet} (w.r.t.~$x$) of~$W$. A~simple
facet of a zone is itself a zone, hence we can look at its simple facets
as~well, and so on iteratively 
we~call them the \emph{facets} of~$W$, and we
say that~$W$ is a facet of~$W$ as~well.
Facets of a non-closed zone are those of its closure.

We note $\LF_x(W)$ (resp. $\UF_x(W)$) the set of \emph{lower simple facets
  w.r.t~$x$} (resp. \emph{upper simple facets w.r.t.~$x$}) of~$W$, i.e. those
obtained from a lower bound $x \ge n$ or $x - y \geq m$ (resp.~from an upper
bound $x \leq n$ or $x-y \leq m$). We~write~$e_x$ for the unit vector in the
direction of the $x$-axis. If~$F$~is a lower (resp.~upper) simple facet of~$W$
w.r.t.~$x$, and $v_F \in F$, we~note $W^F(v_F)$ the set $\{ v \in W \mid
\exists \lambda \in \posreal .\ v - \lambda.e_x = v_F \}$ (resp. $\{ v \in W
\mid \exists \lambda \in \posreal .\ v + \lambda.e_x = v_F \}$), the set of
points of~$W$ that lie above (resp.~below) $v_F$ in the direction of the
$x$-axis. Similarly, we~note $W^F = \bigcup_{v_F \in F} W^F(v_F)$, the~set of
points of~$W$ that lie above (resp.~below)~$F$ in the direction of the
$x$-axis. Note that it holds $\bigcup_{F \in \LF_x(W)} W^F = W$; 
in~case $W$ contains no half-line of the form $v +
\posreal.e_x$ for any $v \in W$, then we also have  $\bigcup_{F
  \in \UF_x(W)} W^F = W$.

\paragraph{Reduction of one variable in the optimization problem.}
We recall that $\zeta$ and $\zeta'$ are affine functions over $Z$ and
$Z'$ respectively.  We note $\zeta_x = \frac{\partial}{\partial x}
\zeta$ and $\zeta_x' = \frac{\partial}{\partial x} \zeta'$ for every
$x \in X$. We note $c$ the constant term of $\zeta$.

Let $x \in X \setminus Y$.  We define $\mathcal{F}_x(Z_Y)$ as
$\LF_x(Z_Y)$ if $\zeta_x \geq 0$, and as $\UF_x(Z_Y)$ if $\zeta_x < 0$.
In the former case, $Z_Y = \bigcup_{F \in \mathcal{F}_x(Z_Y)} Z_Y^F$
always holds.  In~the latter case, recall that $\zeta$ is bounded from
below on~$Z$, so $Z$ contains no half-line $v + \posreal.e_x$ for any $v
\in Z$, otherwise $\zeta$ would take arbitrarily low values along this
half-line.  This ensures that $Z_Y = \bigcup_{F \in
  \mathcal{F}_x(Z_Y)} Z_Y^F$ in both cases.

Now, we~have:
\begin{xalignat}1
S(\calZ,\calZ',Y) &= \sup_{v \in Z_Y\vphantom{Z'_Y}} \pinf_{\substack{v' \in Z'_Y\\v' \equiv_M
    v}} \zeta'(v') - \zeta(v)\notag\\ 
  &= \sup_{u \in \pi_Y(Z_Y)} \sup_{\substack{v \in Z_Y\\\pi_Y(v) =
    u}} \inf_{\substack{v' \in Z_Y'\\\pi_Y(v') = u}} (\zeta'(v')
    - \zeta(v))\notag\\\noalign{\pagebreak[1]}
  &= \sup_{u \in \pi_Y(Z_Y)} \Bigl[\Bigr(\inf_{\substack{v' \in Z_Y'\\\pi_Y(v') = u}}
  \zeta'(v')\Bigr) + \Bigl(\sup_{\substack{v \in Z_Y\\\pi_Y(v) = u}}
    (- \zeta(v))\Bigr)\Bigr] \notag \\
\noalign{\hfill(because the constraint on~$v'$ is now independent
  of~$v$}
\noalign{\hfill and both $\zeta$ and $\zeta'$ are bounded from below)}
  &= \sup_{u \in \pi_Y(Z_Y)} \Bigl[\Bigl(\inf_{\substack{v' \in
  Z_Y'\\\pi_Y(v') = u}} \zeta'(v')\Bigr) - 
  \Bigl(\inf_{\substack{v \in Z_Y\\\pi_Y(v) = u}} \zeta(v)\Bigr)\Bigr]
  \label{eq-dimZY}\\
  &= \max_{F \in \mathcal{F}_x(Z_Y)} \sup_{u \in \pi_Y(Z_Y^F)} \Bigl[\Bigl(\inf_{\substack{v' \in
  Z_Y'\\\pi_Y(v') = u}} \zeta'(v')\Bigr) - \Bigl(\inf_{\substack{v \in Z_Y\\\pi_Y(v) =
  u}} \zeta(v)\Bigr)\Bigr].\notag 
\end{xalignat}

For~$F\in \mathcal{F}_x(Z_Y)$, write~$\gamma_F$ for the extra
constraint defining~$F$ in $\adherence{Z_Y}$ (of~the form $x=n$ or
${x-y=m}$). For~$v \in F$, the~value of~$v(x)$ is entirely determined
by the values $(v(y))_{y \ne x}$ and by the constraint $\gamma_F$:
$v(x)=n$ or $v(x)=v(y)+m$. Thus, $\pi_{X \setminus \{x\}}$ is a
bijection between $F$ and $\pi_{X \setminus \{x\}}(F)$. For $w \in
\pi_{X \setminus \{x\}}(F)$, we note $\iota_F(w)$ the unique point of
$F$ such that $w = \pi_{X \setminus \{x\}}(\iota_F(w))$.

We define a new cost function $\zeta_F$ on $\pi_{X\setminus\{x\}}(F)$
by $\zeta_F(w) = \zeta(\iota_F(w))$. We~remark that $\zeta_F(w) =
\zeta_x\cdot (\iota_F(w))(x) + \sum_{y \in X \setminus \{x\}}
\zeta_y\cdot w(y) + c$. Since $(\iota_F(w))(x)$ is an affine function
of~$w$, so~is~$\zeta_F$.

By definition of $\mathcal F_x(Z_Y)$, we~have $\zeta_F(w) = \inf_{v
  \in Z_Y^F(\iota_F(w))} \zeta(v)$. Now, remark that
$Z_Y^F(\iota_F(w)) = \{v \in Z_Y \mid \pi_{X \setminus \{x\}}(v) =
w\}$, so that $\zeta_F(w) = \inf_{v \in Z_Y,\ \pi_{X \setminus
    \{x\}}(v)=w} \zeta(v)$.
For $u \in \pi_Y(Z_Y^F)$, we can then write:
\begin{xalignat*}1
  \inf_{\substack{v \in Z_Y \\ \pi_Y(v) = u}} \zeta(v) &=
    \inf_{\substack{w \in \pi_{X \setminus \{x\}}(F) \\ \pi_Y(w)=u}}
    \inf_{\substack{v \in Z_Y \\ \pi_{X \setminus \{x\}}(v)=w}} \zeta(v)
  \\
  &=
    \inf_{\substack{w \in \pi_{X \setminus \{x\}}(F) \\ \pi_Y(w)=u}} \zeta_F(w)
\end{xalignat*}

The above being true for every $F \in \mathcal{F}_x(Z_Y)$, we~get:
\[
S(\calZ,\calZ',Y) = \max_{F \in \mathcal{F}_x(Z_Y)} \sup_{u \in \pi_Y(Z_Y^F)}
\Bigl[\Bigl(\inf_{\substack{v' \in Z_Y'\\\pi_Y(v') = u}} \zeta'(v')\Bigr) -
\Bigl(\inf_{\substack{w \in \pi_{X \setminus \{x\}}(F)\\\pi_Y(w) = u}}
\zeta_F(w)\Bigr)\Bigr].
\]
Finally, note that $\adherence{\pi_Y(F)} = \adherence{\pi_Y(\pi_{X
    \setminus \{x\}}(F))} = \adherence{\pi_Y([\pi_{X \setminus
    \{x\}}(F)]_Y)} = \adherence{\pi_Y(Z_Y^F)}$ (because $x\in
X\setminus Y$ and $F\subseteq \adherence{Z_Y}$), so that we conclude:
\[
S(\calZ,\calZ',Y) = \max_{F \in \mathcal{F}_x(Z_Y)} \sup_{u \in
  \pi_Y([\pi_{X \setminus \{x\}}(F)]_Y)}
\Bigl[\Bigl(\inf_{\substack{v' \in Z_Y'\\\pi_Y(v') = u}}
\zeta'(v')\Bigr) - \Bigl(\inf_{\substack{w \in \pi_{X \setminus
      \{x\}}(F)\\\pi_Y(w) = u}} \zeta_F(w)\Bigr)\Bigr].
\]
This can be rewritten as\footnote{Notice that the definition
  of~$S(\calZ,\calZ',Y)$ requires $\calZ$ and~$\calZ'$ to have the
  same dimension (because $\equiv$~does). However, using
  Eq.~\eqref{eq-dimZY}, we~can extend the definition to any zones
  involving at least the clocks of~$Y$.}
\[
S(\calZ,\calZ',Y) = \max_{F \in \mathcal{F}_x(Z_Y)} S((\pi_{X
  \setminus \{x\}}(F),\zeta_F),\calZ',Y)
\]
where $(\pi_{X \setminus \{x\}}(F),\zeta_F)$ is a new priced zone.

\paragraph{Iteration of the construction.}
For each sequence $(x_i)_{1 \le i \le p}$ of distinct elements of~$X\setminus Y$,
we define $\mathcal{F}_{x_1,\dots,x_p}(Z_Y)$ inductively as follows:
\[
\mathcal{F}_{x_1,\dots,x_p}(Z_Y) = \bigcup_{F \in
  \mathcal{F}_{x_1,\dots,x_{p-1}}(Z_Y)} \mathcal{F}_{x_{p}}(\pi_{X
  \setminus \{x_1,\dots,x_{p-1}\}}(F)).
\]
Note that $F \in \mathcal{F}_{x_1,\dots,x_{p-1}}(Z_Y)$ is a zone over
$X \setminus \{x_1,\dots,x_{p-1}\} \supseteq Y$.

By iteratively applying the above construction, we~can compute, for each ${F\in
\mathcal{F}_{x_1,\dots,x_{p-1}}(Z_Y)}$, an~affine function~$\zeta_F$ such that
\begin{equation}
\forall w \in \adherence{\pi_{X \setminus \{x_1,\dots,x_p\}}(F)}.\quad 
\zeta_F(w) = \inf_{\substack{v \in Z_Y \\ \pi_{X \setminus
      \{x_1,\dots,x_p\}}(v)=w}} \zeta(v).
\label{eq-zetainf}
\end{equation}
The previous analysis also entails:
\[
S(\calZ,\calZ',Y)  = \max_{F \in \mathcal{F}_{x_1,\dots,x_p}(Z_Y)} \sup_{u \in
  \pi_Y(F)} \inf_{\substack{v'
    \in Z_Y'\\\pi_Y(v') = u}} \zeta'(v') - \inf_{\substack{w \in
    \pi_{X \setminus \{x_1,\dots,x_p\}}(F)\\\pi_Y(w) = u}} \zeta_F(w).
\]

Then, when $\{x_1,\dots,x_p\} = X\setminus Y$, we~have $X \setminus
\{x_1,\dots,x_p\}=Y$, so that 
\[
\forall u\in \pi_Y(F).\quad
\inf_{\substack{w \in
    \pi_{X \setminus \{x_1,\dots,x_p\}}(F)\\\pi_Y(w) = u}} \zeta_F(w) = 
\inf_{\substack{w \in
    \pi_{Y}(F)\\\pi_Y(w) = u}} \zeta_F(w) =
\zeta_F(u).
\]
It~follows:
\begin{xalignat*}1 S(\calZ,\calZ',Y) &= \max_{F \in
    \mathcal{F}_{X\setminus Y}(Z_Y)} \sup_{u \in \pi_Y(F)}
  \inf_{\substack{v' \in Z_Y'\\\pi_Y(v') =
      u}} \zeta'(v') - \zeta_F(u) \\
  & = \max_{F \in \pi_Y(\mathcal{F}_{X\setminus Y}(Z_Y))} \sup_{u \in F}
  \inf_{\substack{v' \in Z_Y'\\\pi_Y(v') =
      u}} \zeta'(v') - \zeta_F(u) 
\end{xalignat*}
where we write $\calF_{X\setminus Y}(Z_Y)$ for
$\mathcal{F}_{x_1,\dots,x_p}(Z_Y)$. Elements thereof are facets of
zones over $Y \cup \{x_p\}$, and elements of $\pi_Y(\calF_{X\setminus
  Y}(Z_Y))$ are thus zones over $Y$.

\medskip A similar construction of facets for $Z_Y'$ can be performed:
for $x\in X\setminus Y$, we~consider the set $\calF'_x(Z'_Y)$ of upper
or lower (depending on~$\zeta'_x$) facets of~$Z'_Y$
w.r.t.~$x$. For~each~$F'\in \calF'_x(Z'_Y)$, we~define the~affine
function~$\zeta'_{F'}$ over $\pi_{X\setminus \{x\}}(F')$ using the
``inverse'' projection on~$F'$. By~construction of~$\calF'_X(Z'_Y)$,
it~satisfies $\zeta'_{F'}(w')=\inf_{v'\in Z'_Y,\
  \pi_{X\setminus\{x\}}(v')=w'} \zeta'(v')$ for
all~$w'\in\pi_{X\setminus\{x\}}(F')$.  The situation is similar as
previously, and we end up with
\[
\forall u\in \pi_Y(Z^{\prime F'}_Y).\   \inf_{\substack{v' \in Z'_Y \\
    \pi_Y(v') = u}} \zeta'(v') = \inf_{\substack{w' \in \pi_{X
      \setminus \{x\}}(F') \\ \pi_Y(w')=u}} \zeta'_{F'}(w').
\]
It~follows:
\begin{xalignat*}1 S(\calZ,\calZ',Y) &= \pmax_{F \in
    \pi_Y(\mathcal{F}_{X\setminus Y}(Z_Y))} \pmax_{F'\in
    \calF_{x}(Z'_Y)} \sup_{u\in F \cap \pi_Y(Z^{\prime F'}_Y)}
  \inf_{\substack{w' \in \pi_{X \setminus \{x\}}(F') \\ \pi_Y(w')=u}}
  \zeta'_{F'}(w') - \zeta_F(u).
\end{xalignat*}
By iteratively applying the same transformation for all clocks
in~$X\setminus Y$, and using the same notations as above, we
finally~get
\begin{equation}
  S(\calZ,\calZ',Y) = \pmax_{F \in \pi_Y(\calF_{X\setminus Y}(Z_Y))} 
  \pmax_{F' \in\pi_Y(\mathcal{F}'_{X\setminus Y}(Z'_Y))}
  \sup_{u \in F \cap F'} \zeta'_{F'} (u) - \zeta_F(u).
\label{eq-final}
\end{equation}
This concludes the proof of~\Cref{theo:projections-facets}.

\subsection{Effectiveness of $\sqsubseteq$}

\relaxation*

\begin{proof}
  Assume $\zeta$ is not lower-bounded on $Z$ but $\zeta'$ is
  lower-bounded on $Z'$, then it is easy to see that $\mathcal{Z}
  \not\sqsubseteq \mathcal{Z}'$ (picking $v \in Z$ such that $\zeta(v)
  < \inf_{v' \in Z'} \zeta'(v')-1$, it is easy to see that we cannot
  find any $v' \in Z'_Y$, $v' \equiv v$, such that $\zeta'(v') \le
  \zeta(v)+1$).

  Now assume that $\zeta'$ is not lower-bounded on $Z'$.  Let $Y
  \subseteq X$.  If $\zeta'$ is lower-bounded on $Z_Y'$, $Z_Y
  \subseteq Z_Y'$ is decided with~\Cref{theo:projections-facets} or
  with the result above (depending on whether $\zeta$ is lower-bounded
  or not on $Z_Y$).

  Otherwise, there exists a direction $d$ and a valuation $v_0' \in
  Z_Y'$ such that the half-line $v_0' + \posreal \cdot d$ is included in
  $Z_Y'$, and $\zeta'$ decreases strictly along this half-line.  We
  have $d\cdot e_y = 0$ for every $y \in Y$ (otherwise the whole half-line
  would not be in $Z_Y'$), and $d\cdot e_x \geq 0$ for every $x \in X
  \setminus Y$.  Furthermore, as explained below, the half-lines
  $v' + \posreal \cdot d$ are all included in~$Z_Y'$, for and $v' \in Z_Y'$.
  For $v \in Z_Y$, there exists $v' \in Z_Y'$ such that $v \equiv_M
  v'$.  But $d$ is such that the half-line $v' + \posreal \cdot d$ is
  included in~$[v]_M$.  The~affine function~$\zeta'$ takes arbitrarily
  low values along this half-line, so we can always find $v' \equiv_M
  v$ with $\zeta'(v') \leq \zeta(v)$.  This holds even if $\zeta$ is
  not lower-bounded on~$Z_Y$.

\medskip
  We now prove our claim on half-lines: more precisely, we show that for
  any~$v$ and~$v'$ in some zone~$Z$, and any vector~$d$ such that $v+\posreal\cdot d
  \subseteq Z$, it holds $v'+\posreal\cdot d\subseteq Z$.

  We~first characterize the directions in which $Z$ is unbounded.
  Let $C$ be the set of clocks that are unbounded in~$Z$, i.e., those
  clocks~$x$ for which the entry~$(x,0)$ of the normalized DBM of~$Z$
  is~$+\infty$. Let~$D$ be the set of vectors $\sum_{x\in C}
    \alpha_x\cdot e_x$, in~which $\alpha_x\geq 0$ for all~$x\in C$, and
    $\alpha_x \leq \alpha_y$ in case the entry~$(x,y)$ of the normalized DBM
    of~$Z$ is finite.
  Now, pick~$v\in Z$, and $d\in D$. Then assume that for some~$\lambda>0$, 
  $w=v+\lambda \cdot d$ is not in~$Z$: in particular, for some~$\lambda_0>0$,
  $v+\lambda_0\cdot d$ is on the border of~$Z$, hence it satisfies some
  constraint~$x=n$ or $x-y=n$, for which $x\leq n$ or $x-y\bowtie n$ is a tight constraint
  defining~$Z$. The~former case ($x=n$) is impossible, since the fact that the value
  of~$x$ changes in direction~$d$ indicates that $x\in C$, so that $x$ is unbounded.
  Now, assume that $v+\lambda_0\cdot d$ satisfies $x-y=n$ for some tight constraint
  $x-y\bowtie n$ defining~$Z$ ($x-y\leq n$,~say).
  Necessarily at least one of~$x$ and~$y$ must be in~$C$
  (otherwise the value of~$x-y$ would not change in direction~$d$). 
  In~case only one of them were in~$C$, since in~$Z$ it~holds $x\leq y+n$, it
  must be~$y$. But then in the direction of~$d$, the value of~$x-y$ would
  decrease, so it cannot be the case that $x-y\leq n$ gets violated.
  Hence both $x$ and~$y$ are in~$C$, and since $Z$ contains
  constraint~$x-y\leq n$, it~holds~$\alpha_x\leq \alpha_y$. But then again,
  $x-y$ would decrease in the direction of~$d$, so it cannot be the case that
  $x-y\leq n$ gets violated.

  Conversely, if for some direction~$d=\sum_{x\in X} \alpha_x\cdot e_x$
  it~holds $v+\posreal\cdot d\subseteq Z$, then any clock~$x$ for which the
  vector~$e_x$ has a non-zero coefficient in~$d$ is unbounded. Moreover, if
  for two clocks~$x$ and~$y$ the coefficients in~$d$ are such that
  $\alpha_x>\alpha_y$, then the difference $x-y$ cannot be bounded in~$Z$.
  Hence if $x-y$ is bounded, we~must have $\alpha_x\leq \alpha_y$, which
  entails that $d\in D$. Thus $D$ characterizes all directions in which~$Z$ is infinite.

  Now, pick a point~$v$ in~$Z$, hence satisfying the constraints in~$Z$, and a
  direction~$d$ in~$D$. Assume that $v+\lambda\cdot d$ is outsize~$Z$, which
  means that there is some value $\lambda_0>0$ such that
  $v+\lambda_0\cdot d$ is outside of~$Z$. Hence there is some constraint,
  $x\bowtie n$ or~$x-y\bowtie n$,
  defining~$Z$ that holds true of~$v$ and false of~$v+\lambda_0\cdot d$.
  Applying the same arguments as above, we~can prove that this leads to a
  contradiction.  \qed
\end{proof}

\subsection{Complements for~\Cref{sec:righty}}

\Ydownwardclosed*

\begin{proof}
  First notice that $X_{\le M} \subseteq Y$ and $Y \cap X_{>M} =
  \emptyset$, since~$Z_Y\not=\emptyset$.
  Towards a contradiction, assume $Y$ is not downward-closed for~$\preceq_Z$.
  There is $y \in Y$ and $x \in X \setminus Y$ such that $x \preceq_Z y$.
  Clock~$y$ cannot belong to $X_{\le M}$, as otherwise $x$ would also be an
  element of~$X_{\le M}$, thus of~$Y$. Also, $y$ cannot belong to~$X_{>M}$, as 
  otherwise it would not be in~$Y$. Hence, $Z \cap (y \le M(y)) \ne \emptyset$
  and $Z \cap (y>M(y)) \ne \emptyset$. So~it~must be the case that
  $Z \subseteq (x-y \le M(x)-M(y))$. Now $\emptyset \ne Z_Y \subseteq (y \le
  M(y)) \cap (x>M(x))$, which yields a contradiction with $Z_Y \subseteq
  Z \subseteq (x-y \le M(x)-M(y))$. \qed
\end{proof}

\section{Details for~\Cref{sec:termination}}

\preorder*

\begin{proof}
  $\sqsubseteq$ is obviously reflexive; we~prove transitivity.
  Assume $(Z, \zeta) \sqsubseteq (Z', \zeta')$ and $(Z', \zeta')
  \sqsubseteq (Z'', \zeta'')$.  Let $v \in Z$ and $\epsilon > 0$.
  Since $(Z, \zeta) \sqsubseteq (Z', \zeta')$, there exists $v' \in
  Z'$ such that $v \equiv v'$ and $\zeta'(v') \leq \zeta(v) +
  \epsilon/2$.  Since $(Z', \zeta') \sqsubseteq (Z'', \zeta'')$,
  there exists $v'' \in Z''$ such that $v'' \equiv v'$ and
  $\zeta''(v'') \leq \zeta'(v') + \epsilon/2$.  Thus, $v'' \equiv v$
  and $\zeta''(v'') \leq \zeta(v) + \epsilon$.  
\qed
\end{proof}

\wqobis*

\begin{proof}
  Fix $Y \subseteq X$; we first show that $\sqsupseteq$ is a
  well-quasi-order on $M$-bounded-on-$Y$ priced zones.  Pick an
  infinite sequence of $M$-bounded-on-$Y$ priced zones
  $(\mathcal{Z}_i)_{i \ge 0}$, with $\mathcal{Z}_i =
  (Z_i,\zeta_i)$.
  $\pi_Y(Z)$ is a zone, whose corners have integer coordinates,
  and is included in the bounded space $\prod_{x \in Y} [0,M(y)]$.
  Thus, there are finitely many possible values for $\pi_Y(Z)$.
  We cant thus assume  without loss of generality that for every
  $i \ge 0$, $\pi_Y(Z_{i+1}) \subseteq \pi_Y(Z_{i})$.

  From~\Cref{lemma:relaxation}, if there is $i \ge 0$ such that
  $\zeta_i$ is not lower-bounded on $Z_i$, then for every $j \ge i$,
  $\mathcal{Z}_j \sqsubseteq \mathcal{Z}_i$. We~now assume that all
  $\zeta_i$ are lower-bounded on~$Z_i$ by~$\mu$.
  
  With every $i$, we can associate the tuple 
  \[
  \kappa_i = \Big(\min_{\substack{v \in \adherence{Z_i} \\
      \pi_Y(v)=u}} \zeta_i(v)\Big)_{\substack{u \in \prod_{y \in Y}
      [0,M(y)] \\ u\ \text{has integral coord.}}}
  \]
  Each such coordinate has integral value larger than or equal to
  $\mu$, or $+\infty$ (in case the minimum is taken on an empty
  set). There are finitely many $u \in \prod_{y \in Y} [0,M(y)]$ with
  integral coordinates, hence the order on such tuples is a
  well-quasi-order (by Higman's lemma~\cite{higman52}): there exists
  $i<j$ such that $\kappa_i \le \kappa_j$ (pointwise order). Now,
  using \Cref{coro:vertices}, we get that $\mathcal{Z}_j \sqsubseteq
  \mathcal{Z}_i$, which implies the expected result for
  $M$-bounded-on-$Y$ priced zones.

  By application of Higman's lemma again and of~\Cref{coro:decompos},
  we get that $\sqsupseteq$ is a well-quasi-order over priced zones
  which are either not lower-bounded or lower-bounded by $\mu$. \qed
  %
\end{proof}

\subsection{When does a weighted timed automaton generate only priced
  zones either non lower-bounded or lower-bounded by some fixed
  $\mu$?}
\label{app-mu}

We will characterize such weighted timed automata using an extension
of the corner-point abstraction defined in~\cite{BBL08} for
(clock-)bounded weighted timed automata.

Let $\A = (X,L,\ell_0,\Goal,\penalty1000 E,\weight)$ be a weighted
timed automaton. Let $\mathcal{R}$ be the set of regions for $\A$ with
regards to maximal constants $M$. Given $R \in \mathcal{R}$, it is
$M$-bounded on $Y$ for a maximal set $Y$, and a corner of $R$ is a
point of $\mathbb{N}^Y$ which belongs to $\adherence{\pi_Y(R)}$:
somehow, once a clock is larger than the maximal constant, then it
becomes inactive, hence we forget about those clocks when recording
the corners; this allows to have finitely many reachable corners! We
write $\mathcal{R}^\star$ for the set of pairs $(R,\alpha)$, where $R
\in \mathcal{R}$ and $\alpha$ is a corner of $R$.

We construct corner-point abstraction $\A_{\textsf{cp}}$ of $\A$ as
the following weighted graph: its set of states is $L \times
\mathcal{R}^\star$, and its set of transitions is defined by:
\begin{itemize}
\item There are transitions $e = (\ell,(R,\alpha)) \rightarrow
  (\ell,(R,\alpha'))$ where $\alpha$ and $\alpha'$ are distinct
  corners of $R$ and $\alpha'$ is the time successor of $\alpha$ (in
  which case, $\alpha' = \alpha + 1$). We then set $\weight(e') =
  \weight(\ell)$ (intuitively the delay between the corner-points
  $\alpha$ and $\alpha'$ is one time unit).
\item There are transitions $e = (\ell,(R,\alpha)) \rightarrow
  (\ell,(R,\alpha))$ when $R$ is $M$-bounded on $X$ (hence $R$ has a
  single corner), and $\weight(e) = \weight(\ell)$
\item There are transitions $e = (\ell,(R,\alpha)) \rightarrow
  (\ell,(R',\alpha))$ where $R'$ is the time successor region of $R$
  and $\alpha$ is a corner-point associated with both $R$ and $R'$. We
  then set $\weight(e) = 0$(intuitively, there is no delay between the
  corner-point of the two distinct regions).
\item If $e = \ell \xrightarrow{g,Y} \ell'$ is a transition of ${\cal
    A}$, there will be transitions $e' = (\ell,(R,\alpha)) \rightarrow
  (\ell',(R',\alpha'))$ in ${\cal A}_{\textsf{cp}}$ with $R \subseteq
  g$, $R' = [Y \leftarrow 0] R$, $\alpha$ corner-point associated with
  $R$, $\alpha'$ corner associated with $R'$ and $\alpha' = [Y
  \leftarrow 0]\alpha$. We then set $\weight(e') = \weight(e)$.
\end{itemize}
The first three types of transitions are delay transitions whereas the
last are discrete transitions.  It is just an extension of
\cite[Prop. 3, Prop. 5]{BBL08} to get that optimal cost in $\A$
coincides with optimal cost in $\A_{\textsf{cp}}$.

\smallskip We will use this corner-point abstraction to characterize
cases where the cost functions behave safely in the forward
exploration of the weighted timed automaton.

We reuse notations of the previous proof.

\begin{proposition}
  Let $(\ell,\mathcal{Z}_i)_{i \ge 0}$ be an infinite sequence of
  symbolic states computed during the symbolic exploration of $\A$,
  which are lower-bounded, and such that there is $Y \subseteq X$ with
  $\pi_Y((Z_{i+1})_Y) \subseteq \pi_Y((Z_i)_Y)$, but if
  \[
  \kappa_i = \Big(\min_{\substack{v \in \adherence{(Z_i)_Y} \\
      \pi_Y(v)=u}} \zeta_i(v)\Big)_{\substack{u \in \prod_{y \in Y}
      [0,M(y)] \\ u\ \text{has integral coord.}}}
  \]
  then $(\kappa_i)_{i \ge 0}$ is an infinite antichain. Then there is
  a reachable cycle in $\A_{\textsf{cp}}$ with negative cost, which is
  not purely composed of delay transitions.
\end{proposition}

\begin{proof}
  We consider the corner-point abstraction $\A_{\textsf{cp}}$.

  There is $u \in \prod_{y \in Y} [0,M(y)] $ such that $u$ has
  integral coordinates, and such that the sequence $(\kappa_i(u))_{i
    \ge 0}$ is decreasing. Let $v_i \equiv u$ be such that $v_i \in
  \adherence{Z_i}$ and $\zeta_i(v_i)$ yields the minimum for
  $\kappa_i(u)$. Let $\rho_i$ be a run (in the closure of $\A$) that
  leads to $(\ell,v_i)$ with cost $\zeta_i(v_i)$ (notice that we can
  assume it goes through points with integral coordinates).  If we
  project $\rho_i$ onto $\pi_i$ in the corner-point abstraction, state
  $(\ell,v_i)$ coincides with $u$ as a corner of some region
  $R_i$. Gathering all this paths into a tree, we get that this tree
  is infinite, it has finite branching, therefore, by Koenig's lemma,
  there is an infinite branch in the tree, and along that branch,
  there is some state $(\ell,(R,u))$ which is visited infinitely
  often. The value of the cost at each occurrence of this state is
  decreasing, hence there is a reachable cycle with negative cost.
  Assume that this cycle is purely made of delay moves: this means
  that $R$ is the maximal region (all clocks are above their maximal
  constants), and that the transitions correspond to delaying in that
  maximal regions, with a negative cost. This means that we would then
  compute a non lower-bounded priced zone; which is not possible by
  assumption.  So this cycle is not purely made of delay
  transitions. \qed
\end{proof}

The previous proposition expresses the fact that the cost becomes
arbitrarily small, hence decreases to $-\infty$. If the mentioned
cycle is co-reachable, then this means that the foreseen optimal cost
is $-\infty$; on the other hand, if the cycle is not co-reachable,
then the optimal cost might be finite, but this branch of the
exploration will make the algorithm fail to terminate.

\smallskip Conversely, priced zones generated along such a reachable
cycle will form an antichain, but this does not mean
that~\Cref{algo:optreach} will not terminate, since this antichain can
be dominated by some priced zone previously computed in which the cost
function is not lower-bounded.

\putbib[biblio,extra]
\end{bibunit}

\end{document}